\documentclass[a4paper,USenglish,cleveref, autoref, thm-restate, numberwithinsect]{lipics-v2021}

\pdfoutput=1 %
\hideLIPIcs  %

\newboolean{long}
\newcommand{\AppendixSymbol}{\ding{72}}
\setboolean{long}{true} %
\ifthenelse{\boolean{long}}{
	\NewDocumentEnvironment{prooflater}{m}{\begin{proof}}{\end{proof}\ignorespacesafterend}
	\NewDocumentEnvironment{proofsketch}{o +b}{}{\ignorespacesafterend}
	\newcommand{\restateref}[1]{}
	\NewDocumentEnvironment{statelater}{m}{}{}

	\NewDocumentCommand{\onlyShort}{+m}{}
	\NewDocumentCommand{\onlyLong}{+m}{#1}
	\NewDocumentCommand{\shortLong}{+m +m}{#2}
}{
	\NewDocumentEnvironment{prooflater}{m +b}{%
		\expandafter\global\expandafter\def\csname#1\endcsname{\begin{proof}#2\end{proof}}%
	}{\ignorespacesafterend}
	\NewDocumentEnvironment{proofsketch}{O{Proof sketch.}}{\begin{proof}[#1]}{\end{proof}\ignorespacesafterend}
	\usepackage{apptools}

	\newcommand{\restateref}[1]{[\IfAppendix{\AppendixSymbol{}}{\AppendixSymbol{}}]}
	
	\NewDocumentEnvironment{statelater}{m +b}{%
		\expandafter\global\expandafter\def\csname#1\endcsname{#2}%
	}{\ignorespacesafterend}

	\NewDocumentCommand{\onlyShort}{+m}{#1}
	\NewDocumentCommand{\onlyLong}{+m}{}
	\NewDocumentCommand{\shortLong}{+m +m}{#1}
}

\usepackage{graphicx} %
\usepackage{hyperref}
\usepackage{cite}
\usepackage{amsmath,amsfonts,amssymb}
\usepackage{booktabs}

\graphicspath{{./graphics/}}%

\title{Visualizing Treewidth}

\author{Alvin Chiu}{University of California, Irvine, USA}{chiua13@uci.edu}{https://orcid.org/0009-0009-6863-859X}{}

\author{Thomas Depian}{TU Wien, Vienna, Austria}{tdepian@ac.tuwien.ac.at}{https://orcid.org/0009-0003-7498-6271}{Supported by the Vienna Science and Technology Fund (WWTF) Proj.\ No.\ [10.47379/ICT22029].}

\author{David Eppstein}{University of California, Irvine, USA}{eppstein@uci.edu}{}{Supported by NSF grant 2212129.}

\author{Michael T. Goodrich}{University of California, Irvine, USA}{goodrich@uci.edu}{https://orcid.org/0000-0002-8943-191X}{Supported by NSF grant 2212129.}

\author{Martin N{\"o}llenburg}{TU Wien, Vienna, Austria}{noellenburg@ac.tuwien.ac.at}{https://orcid.org/0000-0003-0454-3937}{Supported by the Vienna Science and Technology Fund (WWTF) Proj.\ No.\ [10.47379/ICT22029].}

\authorrunning{A. Chiu, T. Depian, D. Eppstein, M.\,T. Goodrich and M. N{\"o}llenburg} %

\Copyright{Alvin Chiu, Thomas Depian, David Eppstein, Michael T. Goodrich, Martin N{\"o}llenburg}



\relatedversion{A preliminary version of this paper appeared in the proceedings of the 33rd International Symposium on Graph Drawing and Network Visualization (GD 2025). This version is published in the Journal of Graph Algorithms and Applications (JGAA).}
\relatedversiondetails[cite=GD]{Conference Version}{https://doi.org/10.4230/LIPIcs.GD.2025.17}
\relatedversiondetails[cite=JGAA]{Journal Version}{https://doi.org/10.7155/jgaa.v30i2.3273} %
\supplementdetails{Source Code}{https://doi.org/10.17605/OSF.IO/QFZ5V} %

\nolinenumbers %

\EventEditors{Vida Dujmovi\'c and Fabrizio Montecchiani}
\EventNoEds{2}
\EventLongTitle{33rd International Symposium on Graph Drawing and Network Visualization (GD 2025)}
\EventShortTitle{GD 2025}
\EventAcronym{GD}
\EventYear{2025}
\EventDate{September 24--26, 2025}
\EventLocation{Norrk\"{o}ping, Sweden}
\EventLogo{}
\SeriesVolume{357}
\ArticleNo{15}
\usepackage{todonotes}
\usepackage{xspace}
\usepackage{mathtools}
\usepackage{complexity}
\usepackage{pifont}

\onlyShort{\usepackage{flafter}}

\let\oldrestatable\restatable
\def\restatable{\expandafter\oldrestatable}

\NewDocumentCommand{\DecWidth}{s o}{\ensuremath{\IfBooleanTF{#1}{\overline{\omega}}{\omega}\IfNoValueF{#2}{_{#2}}}\xspace}

\NewDocumentCommand{\Drawing}{o}{\ensuremath{\Gamma\IfNoValueF{#1}{_{#1}}}\xspace}

\NewDocumentCommand{\SymbolLinear}{o}{\ensuremath{\textsf{L}\IfNoValueF{#1}{{\textsf{#1}}}}\xspace}
\NewDocumentCommand{\SymbolCircular}{}{\ensuremath{\textsf{C}}\xspace}
\NewDocumentCommand{\SymbolOrbital}{}{\ensuremath{\textsf{O}}\xspace}
\NewDocumentCommand{\SymbolTTCrossing}{}{\ensuremath{\textsf{t}/\textsf{t}}\xspace}
\NewDocumentCommand{\SymbolTECrossing}{}{\ensuremath{\textsf{t}/\textsf{e}}\xspace}
\NewDocumentCommand{\SymbolEECrossing}{}{\ensuremath{\textsf{e}/\textsf{e}}\xspace}

\NewDocumentCommand{\CrossingCountTT}{m}{\ensuremath{\text{cr}_{\textsf{t}/\textsf{t}}(#1)}\xspace}
\NewDocumentCommand{\CrossingCountTE}{m}{\ensuremath{\text{cr}_{\textsf{t}/\textsf{e}}(#1)}\xspace}
\NewDocumentCommand{\CrossingCountEE}{m}{\ensuremath{\text{cr}_{\textsf{e}/\textsf{e}}(#1)}\xspace}
\NewDocumentCommand{\CrossingCount}{m}{\ensuremath{\text{cr}(#1)}\xspace}

\newcommand{\Size}[1]{\ensuremath{\left\vert #1 \right\vert}}
\newcommand{\BigO}[1]{\ensuremath{\mathcal{O}(#1)}}
\newcommand{\probname}[1]{{\normalfont\textsc{#1}}}

\begin{document}

\maketitle

\begin{abstract}
	A \emph{witness drawing} of a graph is a visualization that clearly shows a given property of a graph.
	We study and implement various drawing paradigms for witness drawings to clearly show that graphs have bounded pathwidth or treewidth. 
	Our approach draws the tree decomposition or path decomposition as a tree of bags, with induced subgraphs shown in each bag, and with ``tracks'' for each {vertex of the graph} connecting its copies in multiple bags.
	Within bags, we optimize the vertex layout to avoid crossings of edges and tracks.
	We implement a visualization prototype for crossing minimization using dynamic programming for graphs of small width and heuristic approaches for graphs of larger width.
	{%
		We explore the design space for width-witness drawings and investigate drawing styles that render the subgraph for each bag as an arc diagram with one or two pages or as a circular layout with straight-line edges, and we render tracks either with straight lines or with orbital-radial paths.
		Finally, we report results from an expert evaluation assessing different witness drawing styles.
	}
\end{abstract}

\section{Introduction}
Work in graph drawing can be often viewed as studying paradigms for visualizing a graph,~$G$, to clearly illustrate one or more 
properties of $G$ while also optimizing one or more quality criteria, such as area or number of edge crossings;
see, e.g., \cite{di1999graph,tamassia2013handbook,ellson2002graphviz,junger2012graph,didimo2019survey}.
This viewpoint is related to the concept of a \emph{visual proof}, which is a proof given as a 
visual representation called a \emph{witness} (or visual certificate); see, e.g., \cite{mehlhorn1996checking,MCCONNELL2011119,FKD+.GVP.2025,MNS+.Cgp.1996}.
Ideally, a visual proof should be clear and concise: the property being proven should be immediately 
discernible simply by examining the witness. 
A prominent example is a crossing-free drawing as a witness for planarity. In this paper,
we are interested in visualizing treewidth.

A \emph{tree decomposition} of a graph, $G=(V,E)$, is a tree, $T$, with nodes, $V_1,V_2,\ldots, V_k$, which are 
called \emph{bags}, where 
\begin{enumerate}
	\item Each $V_i$ is a subset of $V$ and $\bigcup_{i=1}^k V_i = V$.
	\item If $v\in V_i$ and $v\in V_j$, then $v$ is in each bag in the (unique) path in $T$ from $V_i$ to $V_j$.
	\item For each edge $(v,w)\in E$ there is a bag that contains both $v$ and $w$.
\end{enumerate} 
The \emph{width} of a tree decomposition is one less than the size of its largest bag, 
and the \emph{treewidth} of a graph is the smallest width
of a tree decomposition.
The \emph{pathwidth} of a graph is the width of a smallest-width tree decomposition whose tree is a path.
See, e.g., \cite{Bod.TCA.2006,MSJ.EST.2019,korach1993tree,eppstein2000diameter,dujmovic2015genus,doi:10.1137/S0097539702416141} %
for more details regarding these topics, including applications in graph drawing and results that
many \NP-hard optimization problems can be solved in polynomial time for graphs with bounded treewidth or pathwidth.

We are interested in paradigms for witness drawings of graphs to certify their bounded pathwidth 
or treewidth.
Our approaches are inspired by an example visual proof from Wikipedia that a graph has treewidth~3; see \Cref{fig:witness-drawing-wikipedia}. 
In this illustration, each bag is drawn as a disk with its member vertices drawn as points
inside the disk.
In addition, for each vertex, $v\in V$, in this drawing,
there is a set of curves, each of which we call a \emph{track}, that connect the different copies of~$v$ in the bags
of the tree decomposition, i.e., the \emph{support} of $v$ in $T$.
The set of tracks in $T$ for a vertex $v\in V$ 
in such a drawing will form a subtree in $T$, which will always be
a path if $T$ is itself a path. Note, however, that even if the support for every vertex is a path, this does
not imply that~$T$ is a path, e.g., as shown in \cref{fig:witness-drawing-wikipedia}.

\begin{figure}[t]
	\centering
	\includegraphics[page=1,scale=.8]{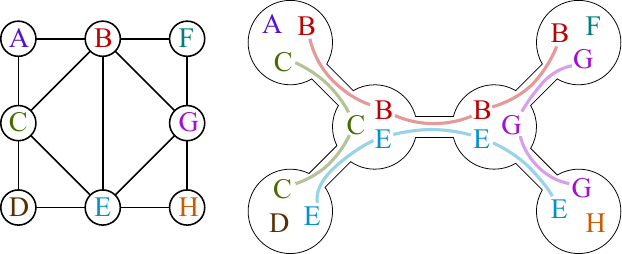}
	\caption{A witness drawing of a graph with treewidth~3. Tracks are shown color coded to illustrate the support
		for each vertex in the graph. Note that this visualization does not 
		draw graph edges inside the bags of the tree decomposition. 
		Public domain images by David Eppstein~\cite{wikitree}.}
	\label{fig:witness-drawing-wikipedia}
\end{figure}

To provide a witness drawing for the pathwidth or treewidth of a graph, $G=(V,E)$, we add one more requirement to our visualization, which is missing from the visualization shown in \Cref{fig:witness-drawing-wikipedia}, 
{%
	namely \emph{information faithfulness}, i.e., that the drawing represents the (entire) graph $G$ and its decomposition $T$.
	In particular, we require} that for each edge, $(v,w)\in E$, we must draw $(v,w)$ as a curve inside each bag $V_i$ such that
$v,w\in V_i$ in the given tree decomposition, $T$, of $G$.
Still, we are interested in such witness drawings that minimize edge and track crossings, 
as in \cref{fig:witness-drawing-wikipedia}.
{%
	Note that we want the witness drawing to be information faithful on its own.
	In particular, when requiring information faithfulness, we do not consider combined drawings of $G$ and $T$ as in \Cref{fig:witness-drawing-wikipedia}.%
}

{%
	To see why this requirement is important, consider a graph $G$ with pathwidth at most \DecWidth.
	Its \emph{interval-width} is bounded by $\DecWidth + 1$, i.e., $G$ is a subgraph of an interval graph $G'$ whose largest clique has size at most $\DecWidth + 1$~\cite{CFK+.PA.2015}.
	Since interval graphs can be represented as closed intervals on the real line, such a representation of $G'$ can be seen as a certificate that $G$ has bounded pathwidth.
	However, we argue that this is not an information faithful witness drawing as it shows $G'$ rather than $G$ and, in particular, does not allow a reconstruction of~$G$ from the (interval) drawing of $G'$.
}

\begin{figure}[t]
	\centering
	\includegraphics[width=\linewidth,page=4]{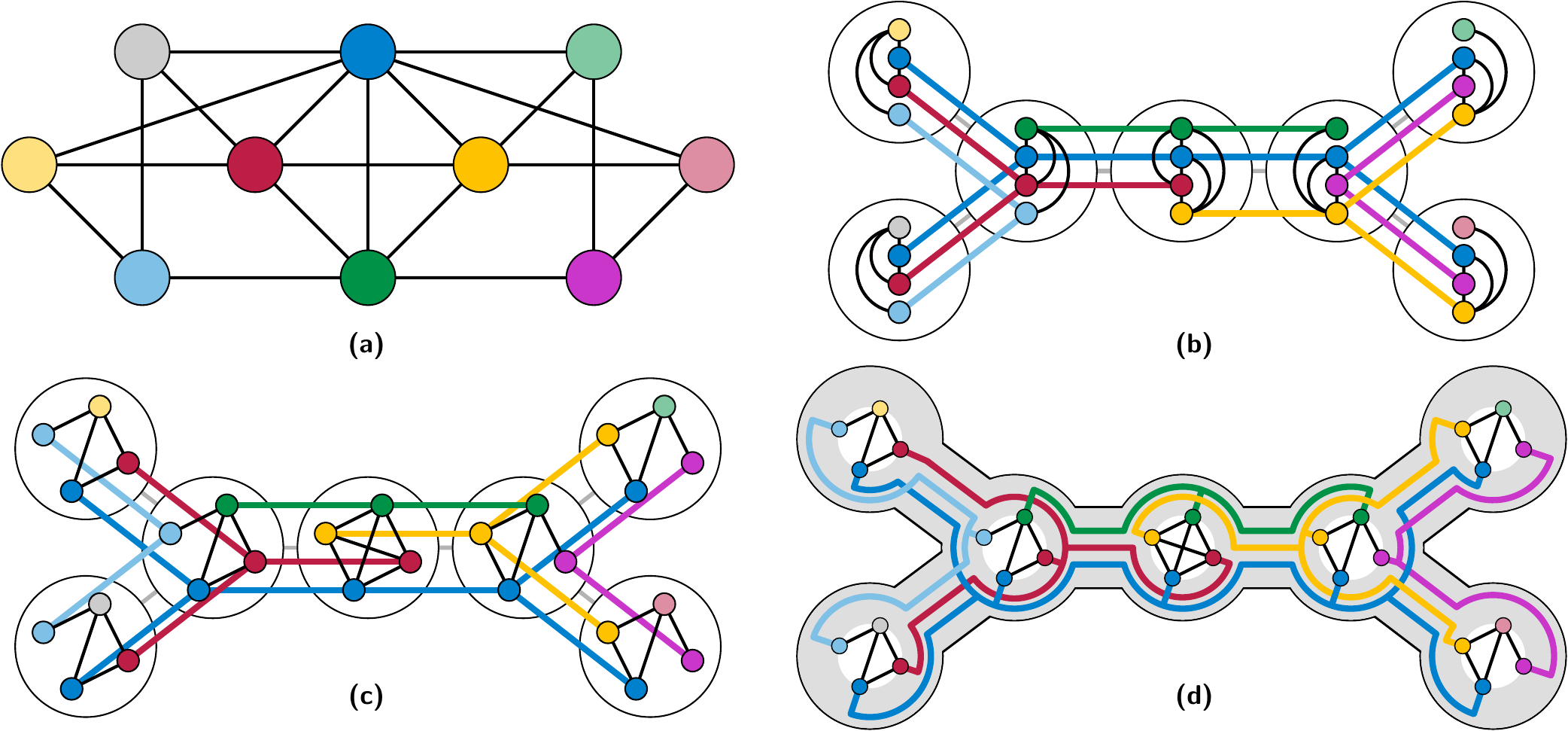}
	\caption{Witness drawings for a graph, $G$, with treewidth~3.
		A standard drawing of $G$ is shown in~\textbf{\textsf{(a)}}.  %
		We show in \textbf{\textsf{(b)}}
		a witness drawing of a tree decomposition, in which the subgraph for each
		bag is drawn as an arc diagram. The two lower witness drawings show
		the subgraph for each bag in a circular layout with straight-line edges. 
		In~\textbf{\textsf{(c)}}, the tracks connecting instances of the same vertex are drawn
		as straight segments, while %
		in~\textbf{\textsf{(d)}} the track edges are routed in orbits surrounding the bag vertices.
		The track edges are colored to match the vertices they represent.
		The drawing %
		in~\textbf{\textsf{(b)}} has 6 track-track crossings, no edge-edge
		crossings, and 6 track-edge crossings.
		The drawing %
		in~\textbf{\textsf{(c)}} has 4 track-track crossings, 1 edge-edge
		crossing, and 13 track-edge crossings. The drawing %
		in~\textbf{\textsf{(d)}} has 12 track-track crossings, 1 edge-edge crossing, and no track-edge crossings.}
	\label{fig:witness-3-decomp}
\end{figure}

\subparagraph*{Our Results.}
In this paper, we systematize and {%
	explore} %
witness drawing {%
	styles} for graphs with bounded treewidth or pathwidth.
{%
	We focus on styles that, given} a graph, $G=(V,E)$, %
draw the vertices in each bag of a tree decomposition, $T$, of $G$ as points in a disk, either in a circular layout or a linear layout.
{%
	More precisely, for} each vertex, $v\in V$, we connect each pair of copies of $v$ in adjacent bags in~$T$ with a curve we call a \emph{track}, ideally with all the tracks color coded with the same color (the color for $v$).
For each bag, $V_i$, in $T$, we draw the edges of the subgraph of $G$ induced by~$V_i$.
In the case of a linear layout of vertices, we draw this subgraph as an arc diagram~\cite{w-dvss-02} (either as one-page or two-page book drawing) and in the case of a circular layout of vertices, we draw this subgraph as a circular layout with straight-line edges~\cite{bb-crcl-04}.
Our optimization goal is to minimize the total number of edge crossings, including track-track crossings, edge-edge crossings, and track-edge crossings.
See Figure~\ref{fig:witness-3-decomp} for three different witness drawing styles.
We design fixed-parameter-tractable linear-time crossing minimization algorithms for graphs with bounded pathwidth or treewidth.
{%
	We implement the above-mentioned algorithm for drawings with a linear layout, suggest a faster heuristic approach, and compare these algorithms.
	Moreover, we conduct an expert evaluation to collect preliminary feedback on the suggested drawing styles and identify the most promising directions for future work.}

\subparagraph*{Related Work.}
Mehlhorn {\it et al.}~\cite{mehlhorn1996checking} and
McConnell %
{\it et al.}~\cite{MCCONNELL2011119} introduce the paradigm of 
\emph{certifying algorithms}, which output their result as well as a concise proof, called a ``witness'',
that shows that the algorithm produced this output correctly. Subsequent to this work, there
has been considerable work on certifying algorithms, including their inclusion in the LEDA and CGAL 
systems~\cite{hoffmann2017two}.
For example, additional work for witness drawings in graph drawing include work on visualizing proximity properties,
such as for Delaunay graphs, Gabriel graphs, and rectangle graphs; see, 
e.g.,~\cite{aronov2014witness,LENHART2023114115,ARONOV2013894,ARONOV2011329,kratsch2006certifying,ADH.WRG.2011,KMMS.Car.2003,LL.MWG.2023}.

In GD 2024, %
the \emph{GraphTrials} framework was introduced,
which involves a visual proof for a graph property that can be used in an interactive proof modeled after a bench trial
before a judge where a prover establishes that a graph has a given property based on its visualization~\cite{forster_et_al:LIPIcs.GD.2024.16}.
{%
	Although %
	the authors did not include pathwidth or treewidth in their GraphTrials framework,
	we feel that our approach to producing witness drawings for graphs with bounded pathwidth or treewidth nevertheless fits into their GraphTrials
	framework. %
	{%
		In particular, we believe that our witness drawings satisfy the three required properties for visual certificates as defined in the framework~\cite{forster_et_al:LIPIcs.GD.2024.16,FKD+.GVP.2025}.
		More concretely, our witness drawings are (i) \emph{unimpeachable}, i.e., information faithful and allow a judge to perceive enough information to validate the property, 
		which is the case since the} visualization inside the bags together with the tracks provides all the information to ensure that the decomposition at hand is valid.
	Moreover, they are {%
		(ii) \emph{checkable}, i.e., it is easy and effortless for a judge to validate the property,} as they facilitate the formation of a mental model. %
	In particular, exploiting the Gestalt principles of \emph{Common Region}, \emph{Proximity}, and \emph{Similarity}~\cite{War.C6S.2013}, the mental model might, on the one hand, clearly partition the drawing into smaller drawings for different subgraphs, while, on the other hand, at the same time linking occurrences of the same vertex in different bags.
	Although the task of checking the certificate is not an instantaneous process even when using this mental model, we argue that it is still faster and easier than not having a witness drawing (and the resulting mental model).
	Finally, a visual certificate must be {%
		(iii) \emph{simple}, i.e., allow a judge to derive all conclusions, in our case bounded pathwidth or treewidth,} %
entirely from the mental model.
Since the above-suspected mental model contains all components of a decomposition, this is indeed the case.

}

We are interested in the natural optimization criterion of minimizing edge
crossings in a witness drawing for a graph with bounded pathwidth or treewidth, including track-track crossings,
edge-edge crossings, and track-edge crossings.
{%
In particular, few track-track crossings can aid in verifying that our decomposition satisfies Property 2, i.e., that the bags containing the vertex $v$ induce a connected subgraph of $T$, for all $v \in V$.
Furthermore, minimizing edge-edge and track-edge crossings improves the visual quality of the witness drawing and can thus support the verification of the (remaining) properties.}
Although we are not aware of any prior work on this criterion, minimizing track-track crossings is
related to minimizing crossings in storyline 
drawings~\cite{dobler_et_al:LIPIcs.GD.2024.31,kostitsyna2015minimizing,gronemann2016crossing,van2017block,van2017computing} 
and metro maps~\cite{bekos2008line,argyriou2010metro,asquith2008ilp,fp-mcmhatc-13,n-iamlc-10,bnuw-miecw-07,pnbh-erwob-12,PNBH.ERO.2012}. 
Minimization of edge-edge crossings within the bags is related to crossing minimization in linear and circular graph layouts~\cite{bb-crcl-04,w-dvss-02,KMN.EEB.2018,gk-icl-06,MKNF.Ncc.1987,BE.CM1.2014,DW.LLG.2004}.

Multiple past works include individual visualizations of tree decompositions. These include drawings of a tree with each bag shown as a set of vertices within each tree node, separate from any drawing of the graph~\cite{Bod.TCA.2006,MSJ.EST.2019}, and drawings of a graph overlaid by bags drawn as regions that surround or connect the vertices~\cite{Epp-NAMS-25}. Separately from their applications in treewidth and structural graph theory, tree decompositions with bags of non-constant size have also been used in other ways: for instance, SPQR decompositions are, essentially, tree decompositions of adhesion~2, whose bags induce 3-connected subgraphs, and these have often been visualized as a tree of 3-connected components~\cite{DiBTam-Algo-96,Mut-ICALP-03}. However, we are unaware of past works studying visualizations of tree decompositions in any systematic way.

\onlyShort{
\smallskip
\noindent
Missing details for statements marked by \AppendixSymbol{} can be found in %
{%
	the full version~\cite{ARXIV}}.
	}

	\section{Preliminaries}
	\label{sec:preliminaries}

	For an integer $p \geq 1$, we use $[p]$ as a shorthand for the set $\{1, 2, \ldots, p\}$.
	Let $G = (V, E)$ be a graph with vertex set $V$ and edge set $E$.
	All considered graphs are simple and undirected. %
	Throughout this paper, we use $n \coloneqq \Size{V}$ and $m \coloneqq \Size{E}$ to refer to the number of vertices and edges of $G$, respectively.
	For a tree (or path) decomposition $T$ of $G$ with bags $V_1, V_2, \ldots, V_k$, we let $\DecWidth[T]$, or simply $\DecWidth$ if $T$ is clear from the context, denote its width.
	Recall $\DecWidth[T] \coloneqq \max_{i \in [k]} \Size{V_i} - 1$.
	We use $\DecWidth* \coloneqq \DecWidth + 1$ as a shorthand for the maximum number of vertices in a bag. %
	Without loss of generality, we root the tree $T$ at an arbitrary bag.
	We use \emph{decomposition} as a synonym for path and tree decompositions.
	Furthermore, we assume {%
that the decomposition $T$ consists of} $k = \BigO{n}$ {%
bags} and that each bag of $T$ has degree at most three. {%
It is well-known that every graph of treewidth \DecWidth admits a certifying tree decomposition with $k = \BigO{n}$ fulfilling this property~\cite{MT.cfi.1992,BK.ECA.1996}}.
In particular, {%
in} \emph{nice tree decompositions}\footnote{%
In nice tree decompositions, $T$ is rooted and every bag $V_i$ is either an \emph{introduce bag}, where $V_i = V_j \cup \{v\}$ with~$V_j$ being the single child of $V_i$ and $v \notin V_j$, a \emph{forget bag}, where $V_i = V_j \setminus \{v\}$ with $V_j$ being the single child of $V_i$ and $v \in V_j$, or a \emph{join bag}, where $V_i = V_j = V_k$ and $V_j$ and $V_k$ are the only two children of $V_i$~\cite{CFK+.PA.2015}.}, which have found attention in the design of parameterized (dynamic programming) algorithms~\cite{CFK+.PA.2015}, {%
every bag has degree at most three {%
	in $T$}.
With slight abuse of notation, we write $V_i \in T$ to indicate that $V_i$ is a bag of $T$.
Furthermore, for a bag $V_i \in T$, we let $G_i \coloneqq G[V_i]$ and $E_i \coloneqq E(G_i)$ denote the subgraph of $G$ induced by $V_i$ and its edge set, respectively.

\section{Exploring the Design Space of Width-Witness Drawings}
\label{sec:taxonomy}

Before we describe how to compute a witness drawing for a decomposition $T$ of $G$, we first take a closer look at the different witness drawings from \Cref{fig:witness-drawing-wikipedia,fig:witness-3-decomp}.
To that end, we first explore the design space of width-witness drawings and categorize different drawing styles based on their features{%
	; recall that we focus on witness drawings for $T$ and not on combined drawings of $T$ and~$G$.}
Afterwards, we elaborate on three particular drawing styles that we discuss in the upcoming sections in more detail. 
{%
	We emphasize that the following exposition, also in view of the plethora of different graph visualization styles, is merely an exploration of the design space.
	The presented drawing styles should be seen as a suggestion and inspiration; styles that we do not mention in this work explicitly can still be sensible.
}

{%
{A} witness drawing $\Drawing$ of a given decomposition $T$ of $G$ consists of a drawing of (i) the tree (or path) $T$, (ii) the subgraph $G_i$ in each bag $V_i \in T$, and (iii) the support for each vertex of~$G$, i.e., the tracks.
{%
	Consequently, also the design space features one dimension for each of the above-mentioned design choices, i.e., how to draw the tree (or path) $T$, the subgraphs~$G_i$ in the bags, and the tracks.
}%
Moreover, when computing a witness drawing $\Drawing$, we seek to optimize for a certain (iv) quality criterion.
Note that the drawing \Drawing might contain crossings, and we differentiate between the following three crossing types; recall also \Cref{fig:witness-3-decomp}:
If two tracks cross, we call it a track-track crossing (or simply \SymbolTTCrossing-crossing), if a track and an edge in a disk cross, it is a track-edge (\SymbolTECrossing) crossing, and if two edges of $G$ cross inside a disk, we call it an edge-edge~(\SymbolEECrossing) crossing.
If neither of these crossings exist, we call $\Drawing$ \emph{planar}.
Note that we do not count crossings with edges of $T$.
The four parameters mentioned above define the design space.
In the following, we explore for each of the four dimensions different choices and discuss their potential influence on the obtained witness drawing.
We refer to \Cref{tab:design-space} for an overview of possible drawing styles.

\begin{table}[tp]
	\centering
	\caption{Visualization of a subset of the design space. In all drawing styles, we draw each bag $V_i \in T$ as a disk $D_i$ of uniform radius and $T$ rightward planar. In a consistent embedding, the embedding inside each bag is identical to that of a drawing of the graph, implying that identical vertices are placed at the same relative position in different disks.
		{%
			The gray region indicates the dedicated track routing area.}
		Combinations {%
			marked} with a cross are not applicable since %
		{%
			the existence of a dedicated interface curve, along which all the vertices in a bag are placed and which separates the track routing area from the bag area cannot be guaranteed.}
		We indicate in green design styles that we study algorithmically in this paper and in blue those that we additionally consider in our expert evaluation in \Cref{sec:user-study}.}
		\label{tab:design-space}
	\medskip
	\begin{tabular}{cccc} 
		\toprule
		& \multicolumn{3}{c}{Tracks}\\
		\cmidrule{2-4}
		Graphs & Straight & {%
			Smooth} Curves & Track Routing Area\\ 
		\midrule 
		\multirow{-8}{*}{\rotatebox{90}{\shortstack{Book\\ Drawing}}} & \includegraphics[page=1]{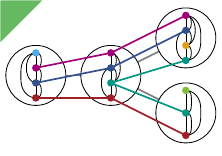} & \includegraphics[page=2]{design-space} & \includegraphics[page=17]{design-space} \\
		\multirow{-8}{*}{\rotatebox{90}{\shortstack{Circular\\ Placement}}} & 
		\includegraphics[page=3]{design-space} & \includegraphics[page=4]{design-space} & \includegraphics[page=5]{design-space} \\
		\multirow{-8}{*}{\rotatebox{90}{\shortstack{Edges Drawn\\ Only Once}}}& 
		\includegraphics[page=6]{design-space} & \includegraphics[page=7]{design-space} & \includegraphics[page=8]{design-space} \\
		\multirow{-8}{*}{\rotatebox{90}{\shortstack{Incidence\\ Matrix}}} & 
		\includegraphics[page=9]{design-space} & \includegraphics[page=10]{design-space} & \includegraphics[page=11]{design-space} \\
		\multirow{-8}{*}{\rotatebox{90}{\shortstack{Consistent\\ Embedding}}} & 
		\includegraphics[page=15]{design-space} & \includegraphics[page=16]{design-space} & \includegraphics[page=17]{design-space} \\[1ex]
		\bottomrule
	\end{tabular}
\end{table}

\subparagraph{Drawing Style of the Tree $\boldsymbol{T}$.}
The first dimension of the design space centers around how the underlying tree (or path) $T$ itself is visualized, i.e., the drawing of its bags and edges.
The bags $V_i \in T$ can be drawn using %
geometric shapes such as disks or squares.
The selected shape should fit the chosen style for the drawing of the subgraphs $G_i$.
To avoid emphasizing some bags over others, {%
	we can draw all bags using the same shape and in the same size.
	Alternatively, if we want to emphasize bags containing many vertices, we can scale the bag proportionally to the number of vertices it contains}.
{%
	As an entirely different visualization style, we can completely omit} the contour of the individual bags{%
	; compare also Figures~\ref{fig:design-space-extra}a and~b.} %
In this case, we rely on the Gestalt principle of proximity~\cite{War.C6S.2013} to group together the drawings of subgraphs into the bags.
For the edges of $T$, we can decide to draw them explicitly, e.g. using straight lines, or omit them and indicate adjacency via the relative positioning of the bags in the plane.
Finally, {%
	there are many different ways to layout the tree $T$ itself~\cite{Schul.TVR.2011}.
	Since we want to enrich the bags with further information, a node-link representation of $T$ should be chosen.
	Moreover, we should draw~$T$ without crossings, e.g., rightward planar or orthogonal planar.
}%
In case of a rooted tree{%
	, e.g., for nice tree decompositions, which are rooted~\cite{CFK+.PA.2015}}, the embedding should clearly reflect the root, e.g., by using a downward- or rightward-planar embedding.

\begin{figure}[t]
	\centering
	\includegraphics[page=1]{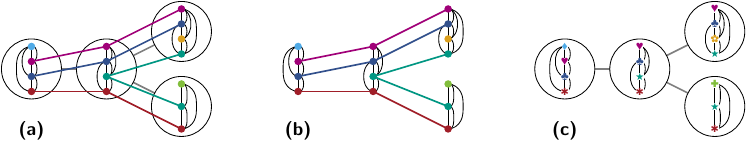}
	\caption{%
		A witness drawing of a tree decomposition where we \textbf{\textsf{(a)}} draw the tree $T$ and the tracks, \textbf{\textsf{(b)}} draw only the tracks, and \textbf{\textsf{(c)}} draw only the tree and use different glyphs instead of the tracks.}
	\label{fig:design-space-extra}
\end{figure}

\subparagraph{Drawing Style of the Subgraphs $\boldsymbol{G_i}$.}
The second dimension of the design space influences the visualization of the subgraph $G_i$ in each bag $V_i \in T$.
While the suitability of a visualization style depends on the structure of both the graph $G_i$ and the tree $T$, it is desirable to represent the graphs consistently across bags.
Moreover, if possible, we should aim for an even distribution of the vertices within each disk.

Consider again \Cref{fig:witness-drawing-wikipedia}.
There, the vertices $V_i$ are arranged along a circle inside the bag.
{%
	If we now draw the edges as straight-line chords of said circle, the resulting drawing is, combinatorial speaking, equivalent to a one-page book drawing of $G_i$~\cite{BK.BTG.1979}}.
Instead of straight chords, we can also draw the edges using arcs, similar to chord diagrams and the chordlink model~\cite{ADM+.CNH.2019}.
Alternatively, we can visualize the adjacencies using drawing styles that are inspired by (outer-)confluent drawings~\cite{DEGM.CDV.2004,DEGM.CDV.2005,EHL+.SCD.2013,EHL+.Scd.2016,FGKN.SOC.2019,FGKN.SOC.2021}.
The former makes more use of the space inside bags and can, for example, highlight highly connected vertices.

Instead of a circular arrangement, we can also arrange the vertices along a vertical or horizontal line within the bag.
This results in a one- or two-page book drawing of $G_i$, where we draw the edges as arcs to either side of the vertex line~\cite{w-dvss-02, KMN.EEB.2018}.
For a one-page book drawing, all arcs are either to the right or to the left of the vertices, but this can change between different bags.
In two-page book drawings, the arcs can be on both sides of the vertices.

For dense graphs, adjacency matrices might be the preferred choice; compare also with the NodeTrix model~\cite{HFM.NHV.2007}.
{%
	We believe that this} representation, as well as book drawings, are particularly suited for path or tree decompositions in which each bag has low degree {%
	in $T$.
	The reason for this is that in these styles, all vertices are arranged along a line inside the bags which gives four (or two) natural attachment directions for the tracks that interfere little with the remaining drawing.
}

Moreover, we can also use force-based approaches~\cite{FR.Gdf.1991} to draw $G_i$ or require that identical subgraphs are drawn consistently across bags.
In essence, any reasonable graph visualization or representation, as long as it exists for each subgraph, falls into this dimension of the design space.
Specific drawing styles that do not exist for all graphs, e.g., plane drawings, cannot be used since we must provide a drawing for the subgraph of each bag.

Finally, recall that a tree decomposition  requires only that for every edge $uv \in G$ there is \emph{some} bag that contains the vertices $u$ and $v$.
Hence, instead of representing $uv$ in all such bags, we can also decide to do so in just one of them.
This might reduce clutter in the individual drawings, but at the cost of information faithfulness: the drawing might then only represent a subgraph of $G_i$.

\subparagraph{Drawing Style of the Tracks.}
The third dimension of the design space concerns the drawing of the tracks that link occurrences of the same vertex in different bags.
The most natural way is to connect them on the shortest path possible using straight lines.
However, we can also try to match the structure of the tree decomposition by using {%
	smooth curves} as in \Cref{fig:witness-drawing-wikipedia} {%
	(in which case the edges of the tree $T$ should accommodate the tracks).
	If tracks are drawn as straight lines, they will be monotone with respect to the tree layout and we enforce the same for smooth curves.%
}
A further option would be to not draw the tracks at all to reduce the clutter of the drawing.
Instead, we could rely on visual encodings, such as color or shape, to convey that vertices in different bags belong together{%
	; compare also Figures~\ref{fig:design-space-extra}a and~c.}
Finally, we might want to restrict the way the tracks can connect to the vertices in a bag.
One way to do this is to introduce a dedicated {%
	(and visually separated)} track routing area inside the bags, which {%
	additionally prevents \SymbolTECrossing-crossings and} helps in reducing the clutter around the drawing of $G_i$, for each $V_i \in T$.
{%
	However, note that we can only use a track routing area if we can define for each bag $V_i \in T$ a continuous interface curve between the track routing area and the drawing of $G_i$ that intersects all vertices but no edges or tracks, i.e., precisely the vertices $V_i$ have access to the interface curve.
}

\subparagraph{Optimization Criteria.}
The final dimension concerns the quality criteria used to evaluate and compare different witness drawings.
Here, classical readability criteria from graph drawing can be applied, e.g., minimizing the number of crossings or the length of the tracks~\cite{KW.DGM.2001,PPBP.GDA.2010}.
When minimizing crossings, we can either treat all types of crossings equally or prioritize certain types over others.
{%
	For example, when tracks are colored, crossings between them could be considered less problematic than crossings between edges of the graph.}
Another possible optimization criterion could be to minimize the bag area covered by tracks.
Thereby, we improve the visibility of the drawings of $G_i$ for each $V_i \in T$.
Moreover, depending on the drawing style of the subgraphs $G_i$ or the tracks, further quality criteria such as the number of bends or the crossing angle could be relevant.
{%
	To identify the priority of these optimization criteria or suggest alternative ones user studies can be conducted.}

{%
	\subparagraph{Further Aspects of Width-Witness Drawings.}
	Observe that most of the above-discussed design choices have an algorithmic impact on the computation of witness drawings.
	However, a comprehensive characterization of the design space of width-witness drawings should also include rendering-related dimensions such as line width or the used color palette; see also the book by Ware~\cite{War.IVP.2013} for further information on these aspects. 
}

\subparagraph{Studied Drawing Styles.}
\Cref{tab:design-space} demonstrates different combinations of drawing styles for the tracks and the subgraphs.
In each of them, we explicitly visualize the underlying tree~$T$ by drawing it rightward-planar, visualizing each bag as a disk, and minimize the number of crossings.
A detailed algorithmic discussion of all drawing styles from \Cref{tab:design-space} is beyond the scope of the paper.
Instead, we focus on a subset of the styles that we deem particularly relevant for an initial study of crossing minimization in width-witness drawings.
On the one hand, these are styles that depict the graph within each bag as a book drawing due to the algorithmic advantages of their combinatorial nature.
On the other hand, these are styles with a circular vertex placement, since they are also frequently used in existing visualizations; see, for example, \Cref{fig:witness-drawing-wikipedia}.
In contrast, we regard the representation of tracks with {%
	smooth} curves mainly as a postprocessing step, and, therefore, do not consider it in our algorithms.

The drawing styles for which we present algorithms in this paper are highlighted in green in \Cref{tab:design-space}.
We illustrate them with a larger example in \Cref{fig:witness-3-decomp} and describe them in more detail in the following.
We use parenthesized letters to refer to each style.
In all studied drawing styles, each bag $V_i \in T$ is represented as a disk $D_i$ of uniform radius. 
Moreover, we draw $T$ explicitly rightward planar, i.e., the leftmost disk corresponds to the root {%
	of $T$} and children of each bag are evenly distributed to their right.
The edges of $T$ are visualized as a straight line; see also \Cref{fig:example-alpha}b.

\begin{figure}[t]
	\centering
	\includegraphics[page=1]{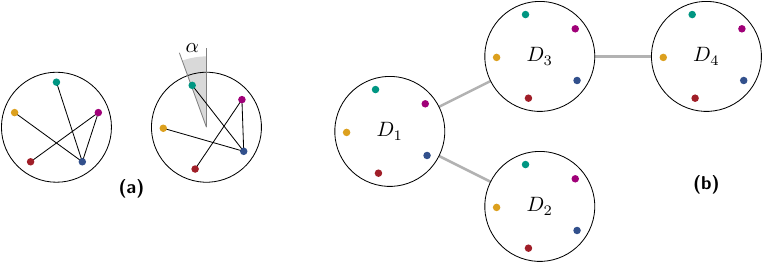}
	\caption{\textbf{\textsf{(a)}} The same circular arrangement of $V_i$ can result in different vertex placements in the disk $D_i$ depending on the angle $\alpha$: left: $\alpha = 0$, right: $\alpha = 20^{\circ}$. \textbf{\textsf{(b)}} Left-to-right drawing of the tree $T$ with four bags. We omit visualizing edges and tracks to maintain readability.}
	\label{fig:example-alpha}
\end{figure}

\subparagraph*{Linear Drawings~(\SymbolLinear).}
A \emph{linear drawing}, as in \Cref{fig:witness-3-decomp}b, consists of a one- or two-page book drawing per bag~\cite{w-dvss-02, KMN.EEB.2018}.
The tracks for each support are straight-line edges connecting the respective vertices.
We use \SymbolLinear[1] and \SymbolLinear[2] to differentiate between one- and two-page book drawings where needed.

\subparagraph*{Circular Drawings~(\SymbolCircular).}
\Cref{fig:witness-3-decomp}c shows a \emph{circular drawing} of $T$.
There, we arrange the vertices $V_i$ along a circle inside the disk $D_i$ and draw the edges $E_i$ as straight-line chords of said circle.
Recall that, from a combinatorial perspective, this is equivalent to a one-page book drawing of~$G_i$~\cite{BK.BTG.1979}.
As for linear drawings, tracks are straight lines.

\subparagraph*{Orbital Drawings~(\SymbolOrbital).}
The last drawing style that we investigate, \emph{orbital drawings}, is inspired by the edge routing in previous drawings with circular vertex placements~\cite{BNRW.TSM.2023,BNT+.BLC.2024}.
It is illustrated in \Cref{fig:witness-3-decomp}d and uses again a circular arrangement of the vertices.
However, we now explicitly introduce a \emph{track-routing area} in each disk, which consists of the remaining part outside{, i.e., around, the drawing of $G_i$}.
The tracks are aligned as orbits inside the track-routing area that can be routed clockwise or counterclockwise around $D_i$. 
Observe that this drawing style has no track-edge crossings {%
	and all track-track crossings occur inside the track-routing area.
	While this can result in a cleaner visualization, forcing tracks to stay inside the track-routing area may clutter drawings of decompositions of larger width.}

\smallskip

Note that for linear drawings, the vertex permutation uniquely defines each vertex' position inside the disk, but for the other two drawing styles this still leaves the rotation of the drawing as a degree of freedeom.
To that end, we let $0 \leq \alpha < 2\pi/\DecWidth*$ denote the starting angle of the first vertex inside each disk, where $\alpha = 0$ corresponds to a placement at twelve o'clock (\Cref{fig:example-alpha}a).
We choose a single $\alpha$ across all disks, and we select it so that no interior of a track passes through a vertex independently of the order of the vertices inside each disk.

\section{Minimizing Crossings in Witness Drawings is \textsf{NP}-hard}
\label{sec:hardness}

{%
	Before we present our fixed-parameter algorithms for computing crossing minimal witness drawings of graphs with bounded path- or treewidth, we justify the focus on decompositions of small width.%
}%

\begin{statelater}{figureHardness}
	\begin{figure}[t]
		\centering
		\includegraphics[page=1]{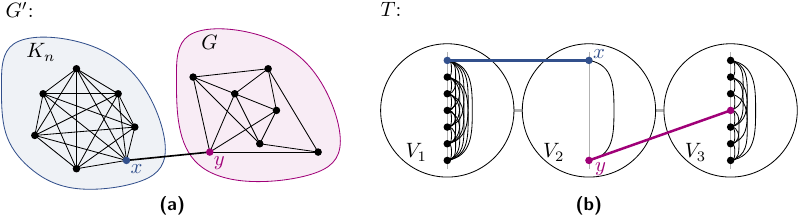}
		\caption{\textbf{\textsf{(a)}} The graph $G'$ with $2n$ vertices and a \textbf{\textsf{(b)}} path decomposition $T$ of $G'$ of width $n - 1$.}
		\label{fig:reduction}
	\end{figure}
\end{statelater}
\begin{restatable}\restateref{thm:hardness}{theorem}{theoremHardness}
	\label{thm:hardness}
	Deciding if a path decomposition $T$ {%
		of a graph $G$} has an \SymbolLinear[1]-drawing with at most $c$ crossings is \NP-hard.
	The same holds true for \SymbolLinear[2]-, \SymbolCircular-, and \SymbolOrbital-drawings.
\end{restatable} 
\begin{prooflater}{ptheoremHardness}
	First, we show the statement for \SymbolLinear[1]-drawings and later argue that it also holds for the other drawing styles.
	We reduce from the \NP-complete problem \probname{One-Page Crossing Minimization}, which asks for a given graph $G$ whether it admits a one-page book drawing with at most $c$ crossings~\cite{MKNF.Ncc.1987}.
	
	Let $(G, c)$ be an instance of \probname{One-Page Crossing Minimization} with $G = (V, E)$ and $n = \Size{V}$.
	We construct the graph $G' = (V', E')$ {%
		and the corresponding path decomposition~$T$ as follows.
		For~$G'$, we take the
	}%
	disjoint union of the two graphs~$G$ and $K_n = (V^n, E^n)$, i.e., a complete graph on $n$ vertices.
	Observe $\Size{V'} = 2n$.
	Let $x$ be an arbitrary vertex of $K_n$ and $y$ an arbitrary vertex of $G$; we add to $G'$ the edge $(x,y)$.
	As $G'$ contains a clique on $n$ vertices, its treewidth (and thus pathwidth) is at least $n - 1$~\cite{CFK+.PA.2015}. %
	{%
		We now construct the path decomposition $T$.
		In particular, we introduce the
	}%
	bags $V_1 = V^n$, $V_2 = \{x,y\}$, and $V_3 = V$, connected in this order.
	{%
		Observe that this shows} that the treewidth (and pathwidth) {%
		of $G'$} is at most $n - 1$.
	{%
		Recall that $G'$ contains, by construction, $K_n$ as subgraph, which has a pathwidth (and treewidth) of $n - 1$.
		Therefore,
	}%
	$T$ is an optimal decomposition of $G'$ with respect to its width.
	It is known that a one-page book drawing of~$K_n$ has %
	$\binom{n}{4}$ crossings~\cite{Sch.GCN.2013}, and, due to symmetry, any linear order of $K_n$'s vertices induces this many crossings.
	We now use this fact to show that $(G, c)$ has a one-page book drawing with at most $c$ crossings if and only if $T$ has an \SymbolLinear[1]-drawing with at most $\binom{n}{4} + c$ crossings.
	
	For the forward direction, let $\prec_G$ be a linear order of $V$ that yields a one-page drawing of~$G$ with at most $c$ crossings. %
	We construct a drawing \Drawing of $T$ by using an arbitrary order $\prec_1$ on $V_1$ that starts with $x$, i.e., $x \prec_1 v$ for every $v \in V^n$, $v \neq x$, and set $x \prec_2 y$ and $\prec_3\ =\ \prec_G$; see also \Cref{fig:reduction}.
	Observe that this drawing has no \SymbolTTCrossing- and \SymbolTECrossing-crossings.
	Thus, there are only \SymbolEECrossing-crossings: $\binom{n}{4}$ in the first, and $c$ in the third bag.
	For the backward direction, assume that $T$ admits an \SymbolLinear[1]-drawing~\Drawing with at most $\binom{n}{4} + c$ crossings.
	Without loss of generality, we can assume that the bags $V_1$, $V_2$, and $V_3$ are drawn in this order from left to right.
	The drawing of $K_n$ in the first bag has $\binom{n}{4}$-many \SymbolEECrossing-crossings.
	Thus, any drawing of $G$ in the third bag can have at most $c$ \SymbolEECrossing-crossings.
	We can use its linear order $\prec_3$ as a witness for $(G, c)$.
	
	The above arguments establish \NP-hardness for \SymbolLinear[1]-drawings.
	One-page book drawings are equivalent to straight-line drawings where the vertices are placed in convex position~\cite{Epp.STG.2002,BK.BTG.1979}.
	Furthermore, between each pair of adjacent bags there is only one track edge.
	Therefore, the same construction and reasoning yields also \NP-hardness for \SymbolCircular- and \SymbolOrbital-drawings independent of the value for $\alpha$.
	For \SymbolLinear[2]-drawings, we can use the known \NP-hardness for determining whether a graph has a two-page book embedding, i.e., a crossing-free two-page book drawing.
	Using the facts that the two-page crossing number of $K_n$ is known~\cite{Sch.GCN.2013} and that book drawings are agnostic to cyclic rotations of the linear order $\prec_G$, we also obtain hardness for \SymbolLinear[2]-drawings.	%
\end{prooflater}
Since a path decomposition is a restricted form of a tree decomposition, \Cref{thm:hardness} also extends to the latter type of decompositions.

\section{\textsf{FPT}-Algorithms for Bounded-Width Decompositions}
\label{sec:algorithms}
{%
	We now focus on the more restrictive {%
		(but arguably also more practical)} setting where we want to compute a witness drawing of a path decomposition $T$ of $G$ of {bounded} width $\DecWidth$.
	In particular, we show that computing crossing-minimal witness drawings is fixed-parameter tractable in $\DecWidth$.
	In \Cref{sec:linear-path,sec:linear-tree}, we first present dynamic programming (DP) algorithms for \SymbolLinear-drawings of path- and tree-decompositions, respectively. 
	In \Cref{sec:circular}, we discuss how to adapt these DP algorithms for \SymbolCircular- and \SymbolOrbital-drawings, i.e., with circular vertex placement.
	
	Throughout this section, we only consider \emph{combinatorial crossings}, i.e., crossings that are purely defined by the relative position of vertices in bags and the bags themselves.
	In particular, we neglect crossings that can be avoided by an appropriate placement of the bags and their included vertices.
}%

\subsection{Linear Drawings for Path Decompositions}
\label{sec:linear-path}

We first consider linear drawings, which visualize each $G_i$ as a book drawing with one or two pages and draw the tracks as straight lines.
Observe that the drawing of $G_i$ is uniquely defined by the linear order $\prec_i$ of $V_i$ and the function $\sigma_i \colon E_i \to \{\ell, r\}$, also called \emph{page assignment}, that assigns each edge to the left or right page, respectively.
Throughout this section, we let the bags of $T$ be ordered as they appear on the path $T$ and assume $V_1$ to be the leftmost bag.

All of our algorithms, including those for \SymbolCircular- and \SymbolOrbital-drawings, are based on the insight that, for each bag $V_i \in T$, the number of crossings involving edges and/or tracks incident to $V_i$ is determined by the drawings of $G_i$ in the disk $D_i$ and the drawings of the induced subgraphs of neighboring bags.
This enables us to employ a dynamic programming (DP) algorithm which processes the path decomposition from left to right and computes a crossing-minimal drawing for the first $i$ bags, for every $i \in [k]$.
Since the framework of our algorithms is identical across all drawing styles (and decomposition types), we first describe its general structure and later discuss how it can be adapted to specific drawing styles and tree decompositions.

\subparagraph*{A General DP Algorithm.}
We let $k$ denote the number of bags in $T$ and recall $\DecWidth* = \DecWidth + 1$.
Let $\Drawing_i$ be a drawing of the graph $G_i$ in the disk $D_i$ for each $i \in [k]$.
Let $\mathcal{G}_i$ denote the set of all possible drawings~$\Gamma_i$ of $G_i$.
All drawing styles admit a combinatorial representation of $\Drawing_i$, which implies that~$\mathcal{G}_i$ is finite for each $i \in [k]$.
Let $\Delta(\Drawing_i)$ denote a drawing of the tracks inside the disk~$D_i$, given a drawing~$\Gamma_i$ of $G_i$, and let $\mathcal{D}(\Gamma_i)$ be the set of all such drawings.
\SymbolLinear- and \SymbolCircular-drawings represent the tracks as straight lines.
Therefore, $\Delta(\Drawing_i)$ is only relevant for \SymbolOrbital-drawings and we leave it undefined for the other two drawing styles. %
For this drawing style, we can again represent~$\Delta(\Drawing_i)$ combinatorially, ensuring that $\mathcal{D}(\Gamma_i)$ is finite.
We define $\mathcal{G}_{\max} \coloneqq \max_{i \in [k]} \Size{\mathcal{G}_i}$ and $\mathcal{D}_{\max} \coloneqq \max_{\Drawing_i \in \mathcal{G}_i, i \in[k]} \Size{\mathcal{D}(\Gamma_i)}$ as the maximum number of drawings of $G_1, \ldots, G_k$ and their corresponding track drawings, respectively.

In our algorithm, we maintain an $\BigO{k \cdot \mathcal{G}_{\max} \cdot \mathcal{D}_{\max}}$-size table~$C$, where $C[i, \Drawing_i, \Delta(\Drawing_i)]$ %
stores the minimum number %
of crossings of a drawing for the bags $V_1, V_2, \ldots, V_i$, where $\Drawing_i$ is the drawing of $G_i$ in $D_i$ and $\Delta(\Drawing_i)$ the drawing of its incident tracks.
We call $S_i = (i, \Drawing_i, \Delta(\Drawing_i))$ a \emph{state} of our DP.
To compute $C[S_i]$, we need to count all \SymbolEECrossing-crossings in $\Drawing_i$ involving edges from $E_i$, and all \SymbolTECrossing-, and \SymbolTTCrossing-crossings involving also edges or tracks incident to vertices from $V_{i - 1}$ for a given state $S_{i - 1}$, if the bag $V_{i - 1}$ exists.
Let $\CrossingCountEE{\cdot}$, $\CrossingCountTE{\cdot,\cdot}$, and $\CrossingCountTT{\cdot,\cdot}$ denote these number of crossings, respectively.
Note that the exact definition of $\CrossingCountEE{\cdot}$, $\CrossingCountTE{\cdot,\cdot}$, and $\CrossingCountTT{\cdot,\cdot}$ depends on the desired drawing style.
For $i > 1$, the number of crossings $\CrossingCount{S_i, S_{i - 1}}$ for two states $S_i = (i, \Drawing_i, \Delta(\Drawing_i))$ and $S_{i - 1} = (i - 1, \Drawing_{i - 1}, \Delta(\Drawing_{i - 1}))$ equals the sum of $\CrossingCountEE{S_i}$, $\CrossingCountTE{S_{i}, S_{i - 1}}$, and $\CrossingCountTT{S_{i}, S_{i - 1}}${%
	; note that $\CrossingCount{S_i, S_{i - 1}}$ can also be defined as a weighted sum to prioritize some crossings over others; recall  \Cref{sec:taxonomy}.}
We can now relate the different states in $C$ to each other as follows, where $S_i = (i, \Drawing_i, \Delta(\Drawing_i))$, $i \in [k]$, $\Drawing_i \in \mathcal{G}_i$, and $\Delta(\Drawing_i) \in \mathcal{D}(\Drawing_i)$: 
\begin{align}
	C[S_i] =
	\begin{cases}
		\CrossingCountEE{S_i} &\text{if}\ i = 1\\
		\min \limits_{\substack{S_{i - 1}\ =\ (i - 1, \Drawing_{i - 1}, \Delta(\Drawing_{i - 1})),\\\Drawing_{i-1} \in \mathcal{G}_{i-1},\ \Delta(\Drawing_{i-1}) \in \mathcal{D}(\Drawing_{i-1})}} C[S_{i - 1}] + \CrossingCount{S_i, S_{i - 1}} &\text{otherwise}
	\end{cases}\label{eq:pathwidth-general-dp-recurrence}
\end{align}
Using induction, we can show that \Cref{eq:pathwidth-general-dp-recurrence} captures the minimum number of crossings.
\begin{restatable}\restateref{lem:pathwidth-general-dp-correctness}{lemma}{lemmaCorrectness}
	\label{lem:pathwidth-general-dp-correctness}
	Let $T$ be a path decomposition of $G$ with $k$ bags.
	For each $i \in [k]$, drawings $\Drawing_i \in \mathcal{G}_i$ and $\Delta(\Drawing_i) \in \mathcal{D}(\Drawing_i)$, the table entry $C[i, \Drawing_i, \Delta(\Drawing_i)]$ equals the minimum number of crossings of a witness drawing for the bags $V_1, \ldots, V_i$ with the drawings $\Drawing_i$ and $\Delta(\Drawing_i)$ for $G_i$ and its incident tracks, respectively. 
\end{restatable}
\begin{prooflater}{plemmaCorrectness}
	We use induction over $i$ to show correctness.
	Throughout the proof, we let $\Drawing_i$ and~$\Delta(\Drawing_i)$ be two arbitrary drawings.
	
	\proofsubparagraph*{Base Case ($\boldsymbol{i = 1}$).}
	For $i = 1$, we are seeking a drawing of the bag $V_1$ only, i.e., a drawing of~$G_1$ in the disk $D_1$.
	Clearly, the minimum number of crossings equals the number of \SymbolEECrossing-crossings in the disk $D_1$, which is precisely $\CrossingCountEE{1, \Drawing_1, \Delta(\Drawing_1)}$; see also the first case in \Cref{eq:pathwidth-general-dp-recurrence}.
	Thus, $C[1, \Drawing_1, \Delta(\Drawing_1)]$ correctly captures the required crossing count.
	
	\proofsubparagraph*{Inductive Step.}
	Let $i > 1$ and note that we are in the second case of \Cref{eq:pathwidth-general-dp-recurrence}.
	Furthermore, let $C[i, \Drawing_i, \Delta(\Drawing_i)] = c$ and assume, for the sake of a contradiction, that there exists a crossing-minimal witness drawing $\Drawing^*$ for the first $i$ bags and with the drawings $\Drawing_i$ and $\Delta(\Drawing_i)$ for the disk $D_i$.
	Let the number of crossings for $\Drawing^*$ be some $c^* < c$.
	We now consider the drawing $\Drawing'$ on the first $i - 1$ bags that is induced by $\Drawing^*$ and denote by $c'$ its number of crossings.
	Let $\Drawing_{i - 1}$ be the drawing of $G_{i - 1}$ and $\Delta(\Drawing_{i - 1})$ be the drawing of its incident tracks in $\Drawing'$.
	Clearly, $c'$ is the minimum number of crossings of a witness drawing for the first $i - 1$ bags of $T$ that uses the drawings $\Drawing_{i - 1}$ and $\Delta(\Drawing_{i - 1})$ for $G_{i - 1}$ and its incident tracks, respectively.
	If this would not be the case, then $\Drawing^*$ would not be crossing-minimal.
	By our inductive hypothesis, we have $C[i - 1, \Drawing_{i - 1}, \Delta(\Drawing_{i - 1})] = c'$.
	Note that we consider in the second case of \Cref{eq:pathwidth-general-dp-recurrence} all possible states for $i - 1$, and in particular the state $(i - 1, \Drawing_{i - 1}, \Delta(\Drawing_{i - 1}))$.
	Since we take among all such states the one that leads to a minimum number of crossings, the fact that we have $C[i, \Drawing_i, \Delta(\Drawing_i)] = c > c^*$ yields a contradiction to the existence of $\Drawing^*$ with $c^*$ crossings, as we would have set $C[i, \Drawing_i, \Delta(\Drawing_i)]$ to its number of crossings $c^*$ otherwise.
	Hence, by induction, the lemma statement follows. 
\end{prooflater}
With \Cref{lem:pathwidth-general-dp-correctness} at hand, we now have all key ingredients for our DP.
It remains to define for each drawing style the representations of $\Drawing_i$ and $\Delta(\Drawing_i)$, if needed, and to implement the crossing counting functions $\CrossingCountEE{\cdot}$, $\CrossingCountTE{\cdot,\cdot}$, and $\CrossingCountTT{\cdot,\cdot}$.
This will be the main task for the remainder of this and the upcoming sections.

\subparagraph*{Two-Page Linear Drawings.}
We now aim to compute a crossing-minimal \SymbolLinear[2]-drawing of a path decomposition $T$.
Recall that a two-page book drawing $\Drawing_i$ of $G_i = (V_i, E_i)$ is uniquely defined by the linear order $\prec_i$ on $V_i$ and the page assignment $\sigma_i \colon E_i \to \{\ell, r\}$.
Thus, we set $\Drawing_i = \langle \prec_i, \sigma_i\rangle$.
As the tracks are drawn using straight lines between the vertices in the respective disks, it is sufficient to store only the drawing for each $G_i$, $i \in [k]$, in our DP table~$C$.
As $\mathcal{G}_{\max} = \BigO{\DecWidth*! \cdot 2^{\DecWidth*^2}}$ holds, %
the size of the table $C$ is in $\BigO{k \cdot \DecWidth*! \cdot 2^{\DecWidth*^2}}$ %
and it only remains to implement the crossing counting functions.
{%
	To this end, recall that we restrict ourselves to combinatorial crossings.}

For \SymbolEECrossing-crossings, we observe that $\CrossingCountEE{\cdot}$ equals the number of edges with alternating endpoints in $\prec_i$ but assigned in $\sigma_i$ to the same page; see \Cref{fig:crossing-counts}a and \Cref{eq:crossing-counts-ee}.
\begin{align}
	\CrossingCountEE{S_i} &\coloneqq \Size{\{(u,v), (a,b) \in E_i \mid \sigma((u,v)) = \sigma((a,b)), u \prec_i a \prec_i v \prec_i b\}}\label{eq:crossing-counts-ee}
\end{align}
To determine the number of \SymbolTECrossing-crossings, it is sufficient to observe that an edge $(u,v) \in E_{i - 1}$ with $u \prec_{i - 1} v$ crosses with a track for the vertex $w \in V_{i - 1} \cap V_i$ if and only if $u \prec_{i-1} w \prec_{i-1} v$ and $\sigma_{i-1}((u,v)) = r$ as visualized in \Cref{fig:crossing-counts}b.
A symmetric observation can be made for edges $(u,v) \in E_i$ with $\sigma_{i}((u,v)) = \ell$, yielding the following function:
\begin{align}
	\CrossingCountTE{S_i, S_{i - 1}} &\coloneqq \sum_{\substack{(u,v) \in E_{i-1}\\\sigma_{i-1}((u,v)) = r}}\Size{\{ w \in V_{i - 1} \cap V_i \mid u \prec_{i-1} w \prec_{i-1} v\}}\label{eq:crossing-counts-te}\\
	&+ \sum_{\substack{(u,v) \in E_i\\\sigma_{i}((u,v)) = \ell}}\Size{\{ w \in V_{i - 1} \cap V_i \mid u \prec_{i} w \prec_{i} v\}}\notag
\end{align}

This leaves us with counting the \SymbolTTCrossing-crossings, for which it is sufficient to count the number of inversions between $\prec_{i-1}$ and $\prec_{i}$; see \Cref{fig:crossing-counts}c and \Cref{eq:crossing-counts-tt}.
\begin{align}
	\CrossingCountTT{S_i, S_{i-1}} &\coloneqq \Size{\{u,v \in V_{i - 1} \cap V_i \mid u \prec_{i-1} v, v \prec_{i} u \}}\label{eq:crossing-counts-tt}
\end{align}
Combining this with the recurrence relation from \Cref{eq:pathwidth-general-dp-recurrence}, we obtain the following:

\begin{figure}
	\centering
	\includegraphics[page=1]{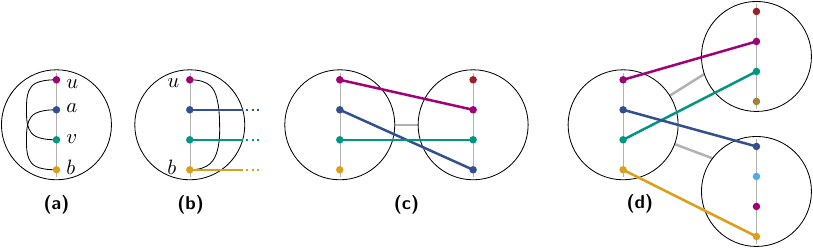}
	\caption{Illustration of different types of crossing: \textbf{\textsf{(a)}} The \SymbolEECrossing-crossing between the two edges $(u,v)$ and $(a, b)$ is enforced by the relative order of their endpoints. \textbf{\textsf{(b)}} The edge $(u,b)$ crosses the blue and green tracks because the respective vertices lie between $u$ and $b$. \textbf{\textsf{(c)}} The blue and green tracks cause a \SymbolTTCrossing-crossing because the relative order among the respective endpoints is inversed. \textbf{\textsf{(d)}} A ``criss-cross'' \SymbolTTCrossing-crossing between blue and green that can occur in tree decompositions.}
	\label{fig:crossing-counts}
\end{figure}

\begin{theorem}
	\label{thm:linear-pathwidth}
	Let $T$ be a path decomposition of $G$ of width $\DecWidth$.
	We can compute a crossing-minimal $\SymbolLinear[2]$-drawing of $T$ in $\BigO{n \cdot (\DecWidth*!)^2 \cdot 4^{\DecWidth*^2} \cdot \DecWidth*^4}$ time.
\end{theorem}
\begin{proof}
	We employ our DP algorithm storing the above-described states to compute a crossing-minimal $\SymbolLinear[2]$-drawing of $T$.
	Recall that the table $C$ has size $\BigO{k \cdot \DecWidth*! \cdot 2^{\DecWidth*^2}}$, where $k = \BigO{n}$ denotes the number of bags in $T$.
	To compute a single state for some $i \in [k]$, we have to access up to $\BigO{\DecWidth*! \cdot 2^{\DecWidth*^2}}$ states for $i - 1$.
	A closer analysis of \Cref{eq:crossing-counts-ee,eq:crossing-counts-te,eq:crossing-counts-tt} reveals that we can evaluate the function $\CrossingCount{\cdot,\cdot}$ in $\BigO{\DecWidth*^4}$ time.
	Thus, we can compute the minimum number of crossings for a given state in $\BigO{\DecWidth*! \cdot 2^{\DecWidth*^2} \cdot \DecWidth*^4}$ time.
	Overall, this amounts to $\BigO{k \cdot (\DecWidth*!)^2 \cdot 4^{\DecWidth*^2} \cdot \DecWidth*^4}$ time to fill the entire table $C$.
	The minimum number of crossings can be obtained by taking the minimum over all $C[k, \cdot, \cdot]$.
	We can use standard backtracking techniques to compute the crossing-minimal $\SymbolLinear[2]$-drawing of $T$.
	The correctness of our algorithm follows directly from the correctness of our recurrence relation, i.e., \Cref{eq:pathwidth-general-dp-recurrence}, established in \Cref{lem:pathwidth-general-dp-correctness}.
\end{proof}
{%
{By} restricting the DP table $C$ and relation from \Cref{eq:pathwidth-general-dp-recurrence}, we can generalize the above-presented DP.
For example, to compute an \SymbolLinear[1]-drawing, we can remove the page assignment from our definition of a drawing $\Drawing_i$ of $G_i$.
Furthermore, %
it could be desired to enforce a consistent linear order across different bags, i.e., forbid \SymbolTTCrossing-crossings.
To that end, we can consider in \Cref{eq:pathwidth-general-dp-recurrence} only those linear orders $\prec'$ that are consistent with $\prec$, i.e., where $\prec\mid_{V_{i - 1}}\ =\ \prec'\mid_{V_i}$.
Below, we summarize the effects on the running time.
\begin{corollary}
	\label{cor:linear-pathwidth}
	Let $T$ be a path decomposition of $G$ of width $\DecWidth$.
	We can compute a crossing-minimal
	\begin{itemize}
		\item $\SymbolLinear[1]$-drawing of $T$ in $\BigO{n \cdot (\DecWidth*!)^2 \cdot \DecWidth*^4}$ time.
		\item $\SymbolLinear[2]$-drawing of $T$ without \SymbolTTCrossing-crossings in $\BigO{n \cdot (\DecWidth*!)^2 \cdot 4^{\DecWidth*^2} \cdot \DecWidth*^4}$ time.
	\end{itemize}
\end{corollary}

\subsection{Linear Drawings for Tree Decompositions}
\label{sec:linear-tree}

Similarly, for graphs of bounded treewidth, we can obtain the minimum number of total crossings for $\SymbolLinear$-drawings via DP. Given a tree decomposition $T$ of $G$, we orient $T$ ``left-to-right'', %
where the root bag is leftmost and children are to the right of the parent. We then obtain a right-to-left ordering of the tree's bags, %
$V_k, V_{k-1}, \dotsc, V_{1}$, %
where $V_1$ is the root bag.

\subparagraph*{DP for Treewidth.}
Our DP formulation in this problem uses the same table as for path decompositions, but a given bag may now have more than one previous bag that can affect its crossings. We assume that the tree decomposition $T$ only has bags of degree up to three. As such, any bag has at most one parent (due to the ``left-to-right'' orientation) and at most two children. We consider three cases to calculate $C[S_i]$ for a bag $V_i$, based on its in-degree $\deg^-(V_i) \in \{0, 1, 2\}$. The cases where $\deg^-(V_i)$ equal zero or one are similar to the two cases presented in the DP for pathwidth, but when  $\deg^-(V_i) = 2$ the crossing function $\CrossingCount{\cdot, \cdot, \cdot}$ depends on the current state and its two child states. In this case, the current bag has two children belonging to its in-degree neighboring set $N^-(V_i)$ whose embedding (which child is above the other) is decided by the DP. Our DP table has the following recurrence relation:
\begin{align}
	C[S_i] =
	\begin{cases}
		\CrossingCountEE{S_i} &\text{if} \deg^-(V_i) = 0\\
		\min \limits_{\substack{S_{x}\ =\ (x, \Drawing_{x}, \Delta(\Drawing_{x}))\\ V_x \in N^-(V_i)}} C[S_{x}] + \CrossingCount{S_i, S_{x}} &\text{if} \deg^-(V_i) = 1 \\
		\min \limits_{\substack{(S_x, S_y) \\\ V_x,V_y \in N^-(V_i), x\neq y}} C[S_{x}] + C[S_{y}] + \CrossingCount{S_i, S_x, S_y} &\text{if} \deg^-(V_i) = 2
	\end{cases}\label{eq:treewidth-general-dp-recurrence}
\end{align} 

\subparagraph*{Two-Page Linear Drawings.} 
The number of $\SymbolEECrossing$-crossings does not change within the parent bag, so that remains the same as in \Cref{eq:crossing-counts-ee}. The number of $\SymbolTECrossing$-crossings can be treated independently for both children $V_x, V_y \in N^-(V_i)$, as there are no such crossings between the children. So \Cref{eq:child-crossing-counts-te} is similar to \Cref{eq:crossing-counts-te} except that we sum over both children.
\begin{align}
	\CrossingCountTE{S_i, S_x, S_y} &\coloneqq \sum_{i' \in \{x,y\}} 
	\Bigg( \sum_{\substack{(u,v) \in E_{i'}\\\sigma_{i'}((u,v)) = \ell}}\Size{\{ w \in V_{i'} \cap V_i \mid u \prec_{i'} w \prec_{i'} v\}} \notag \\
	&+ \sum_{\substack{(u,v) \in E_i\\\sigma_i((u,v)) = r}}\Size{\{ w \in V_{i'} \cap V_i \mid u \prec_i w \prec_i v\}} \Bigg) \label{eq:child-crossing-counts-te}
\end{align}    
{%
	At this point, we want to make the reader aware that we have to maintain a sufficient $x$- and $y$-distance of the bags to each other (or, alternatively, a correct radius for the arcs) to ensure that we only encounter combinatorial crossings in the geometric drawing.
	Otherwise, for example, additional \SymbolTECrossing-crossings can occur between edges and tracks incident to the same vertex.
	We note, however, that there is always a drawing containing precisely the determined combinatorial crossings, although with no guarantee on its resolution.
}

For $\SymbolTTCrossing$-crossings, we must handle $\SymbolTTCrossing$-crossings between children. Without loss of generality, let the child $V_x$ that comes first in the input be above child $V_y$. Then such $\SymbolTTCrossing$-crossings will occur when two tracks from the children ``criss-cross'', where a track from the top child goes below the track from the bottom child as in \Cref{fig:crossing-counts}d, captured by \Cref{eq:child-crossing-counts-tt}. Then we add the $\SymbolTTCrossing$-crossings between parent and child as in \Cref{eq:crossing-counts-tt}, but for both children now:

\begin{align}
	\CrossingCountTT{S_i, S_x, S_y} &\coloneqq \Size{\{(u,v) \in (V_x \cap V_i) \times (V_y \cap V_i) \mid v \prec_{i} u\}} \label{eq:child-crossing-counts-tt} \\ 
	&+ \sum_{i' \in \{x,y\}} \Size{\{u,v \in V_{i'} \cap V_i \mid u \prec_{i'} v, v \prec_i u \}}  \label{eq:tree-crossing-counts-tt}
\end{align}
A similar analysis to \Cref{thm:linear-pathwidth,cor:linear-pathwidth} yields the following results for \SymbolLinear-drawings:%
\begin{theorem}
	\label{thm:treewidth-linear-all-bag-edges}
	Let $T$ be a tree decomposition of $G$ of width $\DecWidth$. We can compute
	a crossing-minimal $\SymbolLinear[2]$-drawing of $T$ in $\BigO{n \cdot (\DecWidth*!)^3 \cdot 8^{\DecWidth*^2} \cdot \DecWidth*^4}$ time.
\end{theorem}
\begin{proof}
	The table $C$ has the same size $\BigO{k \cdot \DecWidth*! \cdot 2^{\DecWidth*^2}}$, with $k = \BigO{n}$, and computing $\CrossingCount{\cdot, \cdot, \cdot}$ still takes $\BigO{\DecWidth*^4}$ time as in the pathwidth case. However, computing a single state must now check all possible pairings of its two child states, yielding $\BigO{(\DecWidth*! \cdot 2^{\DecWidth*^2})^2}$ states in the worst case. So, altogether, the running time will be in $\BigO{n \cdot (\DecWidth*!)^3 \cdot 8^{\DecWidth*^2} \cdot \DecWidth*^4}$.
\end{proof}
\begin{corollary}
	\label{cor:linear-treewidth}
	Let $T$ be a tree decomposition of $G$ of width $\DecWidth$.
	We can compute a crossing-minimal $\SymbolLinear[1]$-drawing of $T$ in $\BigO{n \cdot (\DecWidth*!)^3 \cdot \DecWidth*^4}$.
\end{corollary}

{%
	Finally, we note that our DP algorithm can be extended to decompositions $T$ with bags of arbitrary degree by generalizing \Cref{eq:treewidth-general-dp-recurrence,eq:child-crossing-counts-te,eq:child-crossing-counts-tt}.
	In particular, for \Cref{eq:child-crossing-counts-tt}, we must account for  \SymbolTTCrossing-crossings between every pair of children of $V_i$.
	As stated in \Cref{eq:treewidth-general-dp-recurrence}, we consider all possible combinations of states for $V_i$'s children, and the number of such combinations grows with their number.
	Moreover, observe that the embedding of these children affects the number of \SymbolTTCrossing-crossings, and the number of potential embeddings is itself exponential in the degree {$\text{deg}(V_i)$} of~$V_i$ {%
		in $T$}.
	Therefore, allowing bags with non-constant degree in~$T$ introduces an exponential factor in terms of {%
		the maximum degree of $T$} in the running time of our DP algorithm.
}

\subsection{Drawing Styles with Circular Vertex Placements}
\label{sec:circular}
{%
	We now} discuss how to adapt %
{%
	the above-presented DP algorithms}
to compute crossing-minimal {%
	circular} (\SymbolCircular-) and {%
	orbital} (\SymbolOrbital-) drawings.
{%
	Recall that for \SymbolCircular- and \SymbolOrbital-drawings, the angle parameter~$\alpha$ specifies the starting angle of the first vertex in each disk; recall also \Cref{fig:example-alpha}.
	Since finding a good value of $\alpha$ is an interesting problem in its own right, we consider, in the following, $\alpha$ to be fixed and highlight the steps that depend on the concrete value of $\alpha$.
}

\subparagraph*{Circular Drawings.}
Circular drawings arrange the vertices in every disk $D_i$, $i \in [k]$, on a circle and draw the edges $E_i$ as straight lines.
Thus, the drawing $\Drawing_i$ of $G_i$ is uniquely defined by the placement of the vertices $V_i$ in $D_i$.
Since they are evenly distributed on a circle, it suffices to store in $\Drawing_i$ for each bag $V_i \in T$, the counterclockwise order $\prec_i$ of the vertices $V_i$ as they appear in $D_i$, starting at the twelve o'clock position{%
	, where the angle parameter $\alpha$ specifies the starting angle of the first vertex in the order $\prec_i$ in the disk $D_i$.}
As the tracks are again straight lines, we do not need to store $\Delta(\Drawing_i)$, reducing the size of $C$ to $\BigO{k \cdot \DecWidth*!}$.

When filling the table $C$, we recall that the recurrence relations from \Cref{eq:pathwidth-general-dp-recurrence,eq:treewidth-general-dp-recurrence} include the number of \SymbolEECrossing-, \SymbolTECrossing-, and \SymbolTTCrossing-crossings involving the vertices of a bag $V_i$, $i \in [k]$, and those of its adjacent bags (if they exist).
The number of \SymbolEECrossing-crossings corresponds, due to the equivalence with one-page book drawings, to the number of edge pairs with alternating endpoints in $\prec_i$ and is, therefore, purely combinatorial in this setting.
In contrast, the \SymbolTECrossing-, and \SymbolTTCrossing-crossings are geometric in nature and depend on the concrete positions of the vertices within the disks.
Observe that these positions depend on $\alpha$ and we remark that the choice of~$\alpha$ influences the number of observed crossings; see also 
\Cref{fig:crossing-alpha-radial}. %
Consequently, we equip the crossing counting functions $\CrossingCountTE{\cdot,\cdot,\cdot}$ and $\CrossingCountTT{\cdot,\cdot,\cdot}$ with the parameter~$\alpha$.
Altogether, this is sufficient to compute crossing-minimal \SymbolCircular-drawings.

\begin{figure}
	\centering
	\includegraphics[page=1]{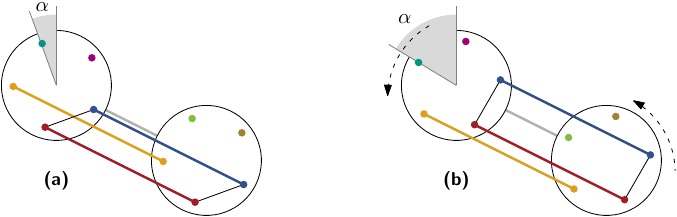}
	\caption{The observed crossings not only depend on the order of the vertices in the disks but also on the choice of $\alpha$ (also indicated with the dashed arrows): The $\SymbolTECrossing$-crossing with the yellow track in \textbf{\textsf{(a)}} is not present in \textbf{\textsf{(b)}} although we use the same linear orders in both visualizations.}
	\label{fig:crossing-alpha-radial}
\end{figure}

\subparagraph*{Orbital Drawings.}
Orbital drawings extend circular drawings by routing the tracks 
as orbits around the vertices; recall \Cref{fig:witness-3-decomp}d.
Since tracks must remain within the track-routing area, every \SymbolOrbital-drawing is free of \SymbolTECrossing-crossings.
However, unlike in the other two drawing styles, the drawing of tracks within a disk is no longer uniquely determined by the placement of the vertices.
In particular, tracks can be assigned to different orbits and may orbit the drawing~$\Drawing_i$ of $G_i$ either clockwise or counterclockwise.
Thus, while it still suffices to store a linear order $\prec_i$ of $V_i$ for $\Drawing_i$, representing $\Delta(\Drawing_i)$ now requires additional care.

We model $\Delta(\Drawing_i)$ with two parts.
First, %
to model the assignment of tracks to orbits, we introduce an \emph{orbit assignment function} $\lambda_i \colon V_i \to [\DecWidth*]$, where $\lambda_i(v)$ specifies the orbit used for the tracks of vertex $v$, numbered
from the center of the disk $D_i$ outward; see \Cref{fig:model-orbits}.
Two tracks for different vertices $u,v \in V_i$ may share the same orbit if they do not intersect.
Otherwise, we require $\lambda_i(u) \neq~\lambda_i(v)$ %
to guarantee that no two tracks for different vertices share the same orbit simultaneously.
Second, to model the direction in which the tracks orbit, i.e., clockwise or counterclockwise, 
we consider %
the (up to two) children $V_x$ and $V_y$ and the parent $V_q$ of the bag $V_i$ in the tree decomposition $T$, if they exist.
We define a \emph{direction function} $\delta_i\colon V_i \times \{x,y,q\} \to \{\text{cw}, \text{ccw}\}$ which, for each vertex $v \in V_i$, captures the direction (clockwise ($\text{cw}$) or counterclockwise ($\text{ccw}$)) of the track from~$V_i$ to each adjacent bag.
\Cref{fig:model-orbits} illustrates the semantic of $\delta_i$ and underlines the importance of storing this information separately for each adjacent bag.
Thus, for a given drawing~$\Drawing_i$ of $G_i$, we set $\Delta(\Drawing_i) = (\lambda_i, \delta_i)$.
This increases the size of the %
DP table $C$ %
to $\BigO{k \cdot \DecWidth*! \cdot {\left(8\DecWidth*\right)}^{\DecWidth*}}$.

\begin{figure}
	\centering
	\includegraphics[page=1]{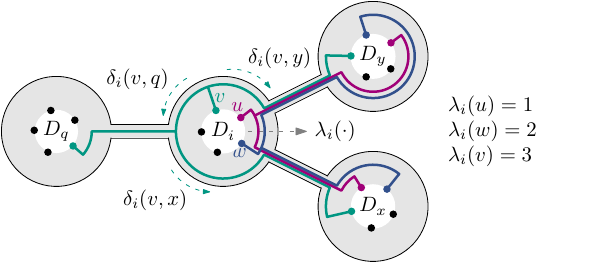}
	\caption{Visualization of a (partial) \SymbolOrbital-drawing and the information that we store in the table $C$.
		The track routing area is indicated in gray.
		Black vertices represent arbitrary other vertices. In this drawing, we have $\lambda_i(u) = 1$, $\lambda_i(w) = 2$, $\lambda_i(v) = 3$, $\delta_i(v,y) = \text{cw}$, and $\delta_i(v,q) = \delta_i(v,x) = \text{ccw}$. Note that these orientations are required to minimize the \SymbolTTCrossing-crossings involving the tracks for vertex $v$.}
	\label{fig:model-orbits}
\end{figure}

Regarding the computation of the number of crossings, we observe that the function $\CrossingCountEE{\cdot}$ from \SymbolCircular-drawings can be reused, and that \SymbolTECrossing-crossings do not occur in this drawing style.
\SymbolTTCrossing-crossings involving a track for a vertex $v \in V_i$ occur in the following two cases.
First, if $v$ also appears in an adjacent bag, and a track for some vertex $u \in V_i$ with $\lambda_i(u) < \lambda_i(v)$ passes by the position of $v$ in $\Drawing_i$, then these two tracks cross: The track for $v$ must cross the one of $u$ to reach its orbit; see \Cref{fig:orbit-counting}a.
Second, for each adjacent bag $V_j$ with $v \in V_j$, we count the number of vertices $u \in V_i$ such that a track for $u$ blocks the path from $D_i$ to~$D_j$, see \Cref{fig:orbit-counting}b.
In both cases, the crossings can be determined from the starting angle $\alpha$, the order of the children, and the stored information, i.e., $\Drawing_i$ and $\Delta(\Drawing_i)$, for the disk $D_i$ of~$V_i$ and the disks of its adjacent bags.
Hence, the function $\CrossingCountTT{\cdot,\cdot,\cdot}$ still only depends on the states for $V_i$ and its child bags.
Finally, since tracks to occurrences of the same vertex $v \in V_i$ in different bags can overlap, we must avoid double-counting in such cases; see again \Cref{fig:orbit-counting}a.
\begin{figure}
	\centering
	\includegraphics[page=1]{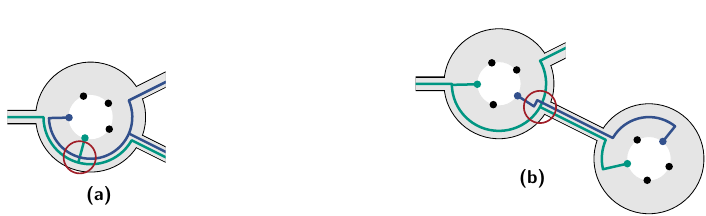}
	\caption{The two cases how \SymbolTTCrossing-crossings in \SymbolOrbital-drawings can arise:
		In \textbf{\textsf{(a)}}, the blue tracks orbit inside the green ones, thus causing a crossing when they pass the green vertex. \SymbolTTCrossing-crossings are highlighted with the red circles.
		The crossing in \textbf{\textsf{(b)}} occurs because the green track blocks the visibility to the lower-right disk for the blue track. 
		Observe that there is only one \SymbolTTCrossing-crossing in \textbf{\textsf{(a)}}, although there is technically one blue track from the upper- and one from the lower-right disk.
	}
	\label{fig:orbit-counting}
\end{figure}
We summarize below the findings of this section. Recall that $k = \BigO{n}$ holds.
\begin{restatable}\restateref{thm:decomposition-circular}{theorem}{theoremCircularDrawings}
	\label{thm:decomposition-circular}
	Let $T$ be a decomposition of $G$ of width $\DecWidth$.
	We can compute a crossing-minimal 
	\begin{itemize}
		\item $\SymbolCircular$-drawing of $T$ in $\BigO{n \cdot (\DecWidth*!)^2 \cdot \DecWidth*^4}$ time if $T$ is a path decomposition.
		\item $\SymbolCircular$-drawing of $T$ in $\BigO{n \cdot (\DecWidth*!)^3 \cdot \DecWidth*^4}$ time if $T$ is a tree decomposition.
		\item $\SymbolOrbital$-drawing of $T$ in $\BigO{n \cdot (\DecWidth*!)^2 \cdot {\left(4\DecWidth*\right)}^{2\DecWidth*} \cdot \DecWidth*^4}$ time if $T$ is a path decomposition.
		\item $\SymbolOrbital$-drawing of $T$ in $\BigO{n \cdot (\DecWidth*!)^3 \cdot {\left(8\DecWidth*\right)}^{3\DecWidth*} \cdot \DecWidth*^4}$ time if $T$ is a tree decomposition.
	\end{itemize}
\end{restatable}
\begin{prooflater}{ptheoremCircularDrawings}
	Let $k$ denote the number of bags of $T$.
	Recall $k = \BigO{n}$ and that our dynamic program (DP) traverses the decomposition $T$ and considers every bag $V_i \in T$, %
	and all the possible drawings $\Drawing_i$ and $\Delta(\Drawing_i)$ of $G_i$ and its incident tracks, respectively.
	For each of them, 
	all possible combinations with %
	states of the two children of $V_i$, if they exist, or $V_{i - 1}$, in case $T$ is a path decomposition.
	\Cref{lem:pathwidth-general-dp-correctness} and the arguments in \Cref{sec:linear-path,sec:linear-tree} ensure correctness of this DP for our definitions of states in the DP.
	Therefore, we focus for the remainder of the proof on the claimed running times and discuss $\SymbolCircular$- and $\SymbolOrbital$-drawings separately.
	
	\proofsubparagraph*{$\SymbolCircular$-drawings.}
	Recall that a drawing $\Drawing_i$ for $G_i$, $i \in [k]$ can, %
	thanks to the angle parameter~$\alpha$, 
	be defined by the counterclockwise order $\prec_i$ of the vertices $V_i$ starting at the twelve o'clock position. %
	As the tracks are straight lines, the size of the DP-table $C$ is  $\BigO{k \cdot \DecWidth*!}$.
	The position for each vertex can be computed in \BigO{1} time.
	Thus, the evaluation of the crossing counting function $\CrossingCount{\cdot, \cdot}$ still requires $\BigO{{\DecWidth*}^4}$ time as for \SymbolLinear-drawings.
	For computing one state of the DP, we must check all $\BigO{\left(\DecWidth*!\right)^2}$ combinations of the up to two children states.
	This yields the claimed running for tree decompositions.
	For path decompositions, only $\BigO{\DecWidth*!}$ states need to be evaluated as each bag has degree at most two {%
		in $T$}.
	
	\proofsubparagraph*{$\SymbolOrbital$-drawings.}
	Recall that for \SymbolOrbital-drawings, we need to specify in addition the drawing of the tracks inside the disk $D_i$, $i \in [k]$, using an orbit assignment function $\lambda_i$ and a direction function $\delta_i$.
	As $\mathcal{G}_{\max} = \BigO{\DecWidth*!}$ and $\mathcal{D}_{\max} = \BigO{\DecWidth*^{\DecWidth*} \cdot 2^{3\DecWidth*}}$, this results in a DP-table $C$ of size $\BigO{k \cdot \DecWidth*! \cdot {\left(8\DecWidth*\right)}^{\DecWidth*}}$.
	As for \SymbolCircular-drawings, we can compute the position for each vertex in \BigO{1} time, which yields the drawing~$\Drawing_i$ of $G_i$.
	Together with $\Delta(\Drawing_i) = (\lambda_i, \delta_i)$, 
	this is sufficient to describe the drawing of the tracks in the disk $D_i$.
	Despite not having \SymbolTECrossing-crossings, computing the number $\CrossingCount{\cdot,\cdot}$ of crossing still takes $\BigO{{\DecWidth*}^4}$ time.
	Since we have to evaluate for each state all combinations with the $\BigO{\DecWidth*! \cdot {\left(8\DecWidth*\right)}^{\DecWidth*}}$-many states for each of its up to two children, the claimed running time for tree decompositions follows.
	If $T$ is a path decomposition, the fact that {%
		$T$ has} degree at most two reduces the size of the table~$C$ to $\BigO{k \cdot \DecWidth*! \cdot {\left(4\DecWidth*\right)}^{\DecWidth*}}$ and the evaluation time for a single state to $\BigO{\left(\DecWidth*!\right) \cdot {\left(4\DecWidth*\right)}^{\DecWidth*} \cdot \DecWidth*^4}$.
\end{prooflater}

\section{Heuristic Approaches for Linear Witness Drawings}
\label{sec:heuristics}
A proof-of-concept implementation of our DP from \Cref{sec:linear-tree} for \SymbolLinear[2]-drawings of tree decompositions required over 8 minutes to terminate even for small decompositions of width $\DecWidth = 3$ with four bags.
This is no surprise given the running time bounds established in \Cref{thm:treewidth-linear-all-bag-edges}.
For tree decompositions of larger width, %
this became too large in practice with a single state needing to check up to $(5!\cdot 2^{5^2})^2 \approx 10^{19}$ possible configurations for $\DecWidth = 4$.

Consequently, we do not want to restrict our attention to exact, but slow, algorithms, but also explore heuristics to efficiently compute drawings with a small number of crossings, but without any formal guarantee on optimality.
In this section, we describe three heuristics for \SymbolLinear[2]-drawings.
{%
	In \Cref{sec:evaluation}, we will} %
evaluate their performance {%
	and provide sample drawings computed by them.} %

\subparagraph*{Global Book Drawing Heuristic.}
Since we visualize in an \SymbolLinear[2]-drawing each graph $G_i$, $i \in k$, as a two-page book drawing of $G_i$, our first heuristic computes such a drawing $\langle \prec, \sigma \rangle$ for the entire graph $G$.
After that, we set for each disk $D_i$ $\prec_i\ =\ \prec\mid_{V_i}$ and $\sigma_i(e) = \sigma(e)$ for every edge $e \in E_i$, i.e., we project the drawing $\langle \prec, \sigma \rangle$ down into each bag.
If $T$ is a tree decomposition, we additionally perform a bottom-up traversal of $T$, choosing for each internal bag $V_i$ the embedding of its children that yields the fewer \SymbolTECrossing- and \SymbolTTCrossing-crossings locally.

Since computing a crossing-minimal one-page book drawing for a graph is an \NP-hard problem~\cite{MKNF.Ncc.1987}, we already employ a heuristic for this step.
Klawitter, Mchedlidze, and Nöllenburg~\cite{KMN.EEB.2018} 
performed an extensive experimentally evaluation on book drawing heuristics.
In their conclusion, they recommended the heuristic \texttt{conGreedy+}, which places one vertex after the other on the spine, selecting the next vertex based on the number of already placed neighbors, and extends the spine order and page assignment greedily.
While for two-page book drawings of $k$-trees, which are maximal graphs of treewidth $k$, it was slightly outperformed by other heuristics, it performed best overall and in particular for graphs with a sub-quadratic number of edges.
The heuristic has a running time of \BigO{m^2n}~\cite{KMN.EEB.2018} %
and we call it \texttt{Global conGreedy+}.

\subparagraph*{Local Book Drawing Heuristic.}
Computing a book drawing for the whole graph $G$ and then applying it to each disk $D_i$, $i \in k$, has two main disadvantages.
On the one hand, we usually do not visualize the entire graph $G$ in a single bag, but several induced subgraphs $G_i$ separately.
Therefore, the drawing $\langle \prec_i, \sigma_i \rangle$ for $G_i$ might contain unnecessary crossings.
On the other hand, the above-described heuristic is, apart from the embedding step, unaware of \SymbolTECrossing- and \SymbolTTCrossing-crossings.
This, in particular, leads to situations where swapping the page assignment of one edge $e \in E_i$ or the entire edge set $E_i$ can dramatically reduce the number of \SymbolTECrossing-crossings.
The latter occurs especially at the root and leave bags of $T$.

Therefore, we now employ the algorithm \texttt{conGreedy+} in each bag $V_i$ separately during a top-down traversal of $T$, thus computing for each $G_i$ a book drawing from scratch.
Since the drawing for the parent of a bag $V_i$ has been determined, we can take the potential \SymbolTTCrossing- and \SymbolTECrossing-crossings into account. {%
	More precisely, when selecting the spine position of a vertex $v \in V_i$ and the page assignment for its incident edges in $G_i$ we also compute the number of \SymbolTECrossing- and \SymbolTTCrossing-crossings that this causes based on the drawing of the parent, similar to \Cref{eq:child-crossing-counts-te,eq:child-crossing-counts-tt}}.
Furthermore, we forward this information to its children when
considering their %
two possible embeddings %
if necessary.
We call this heuristic \texttt{Local conGreedy+} and, %
since the size of each $G_i$ is bounded, its running time is $\BigO{n \cdot \DecWidth*^5}$.

\subparagraph*{Local Search.}
In addition to the two above-mentioned heuristics, we also perform a local search to reduce the number of crossings.
Given an \SymbolLinear[2]-drawing of $T$, we traverse $T$ bottom up.
For each bag $V_i \in T$, we perform one or several of the following movements. %
\begin{description}
	\item[Vertex-Swap:] Take two vertices $u, v \in V_i$ and swap their position on the spine $\prec_i$.
	\item[Edge-Swap:] Take two edges $e_1, e_2 \in E_i$ with $\sigma_i(e_1) \neq \sigma_i(e_2)$ and swap their page assignment.
	\item[Edge-Flip:] Take an edge $e \in E_i$ and assign it to the other page.
	\item[Embedding-Flip:] Flip the embedding of the subtree $V_i$, i.e., the order of its children.
\end{description}
For a given bag $V_i \in T$, we perform a hill-climbing approach and exhaustively apply above movements until no further improvement can be made in the drawing $\langle \prec_i, \sigma_i \rangle$.
Then, we proceed to the next bag $V_j \in T$.
Note that when evaluating the number of crossings, we keep the drawings in the remaining bags fixed.
Thus, after a bottom-up traversal of $T$, we perform a second traversal, this time top-down.
Afterwards, or after a predetermined period of time, we return the best solution found so far.

\section{Experimental Evaluation}
\label{sec:evaluation}
We implemented the dynamic program for \SymbolLinear[2]-crossings as described in \Cref{sec:linear-tree} as well as all of the heuristics presented in \Cref{sec:heuristics} in \texttt{Python 3.6.9}.
All algorithms run on a system with Intel Xeon E5-2640 v4 10-core processors at 2.40 GHz.
We set hard limits for memory and time of 96 GB and fifteen minutes, respectively, but memory was not a limitation for our algorithms.
{%
	In particular, even on the largest instances, our dynamic program used less than 2 GB of memory on average.
}%
If we improve a solution from a heuristic using local search, we indicate this using ``\texttt{\& LS}''.
Also the local search has a maximum running time of fifteen minutes.
The computed drawings are visualized using \texttt{d3.js} on an independent system and this final step is not considered here.
The instances are based on a public repository containing graphs extracted from ``Sage'' together with a tree decomposition of them.\footnote{%
	The original dataset is available at \url{https://github.com/freetdi/named-graphs}.
	The filtered dataset, source code, and evaluation code can be found on \href{https://doi.org/10.17605/OSF.IO/QFZ5V}{OSF}~\cite{SUPPLEMENTARY_MATERIAL}.
	The {%
		(interactive)} application can also be accessed online at \url{https://www.ac.tuwien.ac.at/projects/visualizing-treewidth/}.
}
Out of the 150 graphs, 55 of them were equipped with a {%
	(tree or path)} decomposition $T$ where each bag{%
	, i.e., node of the decomposition $T$,} has degree at most three {%
	in~$T$} and we used these in our experiments.
We have no information on how these decompositions were computed or if they are optimal.
The width of the decompositions ranges between one and 53 and they have between one and 143 bags.
For 60\% of the instances, both parameters were simultaneously below ten. %
The graphs have between four and 143 vertices and between six and 1,008 edges, with a mean density, i.e., $2m / (n(n - 1))$, of $0.37$.

\begin{figure}
	\centering
	\includegraphics[page=1]{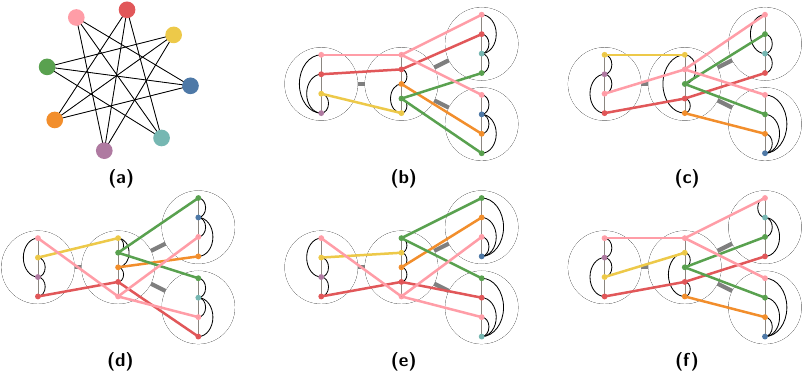}
	\caption{Comparison of the drawings computed by our algorithms for the decomposition of the \emph{Wagner} graph; we indicate the number of crossings in the parenthesis.
		\textbf{\textsf{(a)}} Force-based drawing of the graph, \textbf{\textsf{(b)}} DP (3), \textbf{\textsf{(c)}} \texttt{Global conGreedy+} (8), \textbf{\textsf{(d)}} \texttt{Local conGreedy+} (9), \textbf{\textsf{(e)}} \texttt{Local conGreedy+ \& LS} (5), and \textbf{\textsf{(f)}} \texttt{Global conGreedy+ \& LS} (4).}
	\label{fig:sample-wagner}
\end{figure}

\begin{figure}
	\centering
	\includegraphics[page=1]{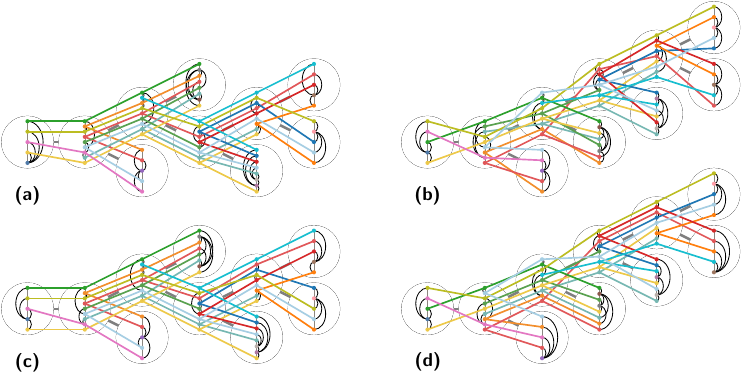}
	\caption{
		Comparison of the drawings computed by our heuristics for the \emph{Brinkmann} graph\protect\footnotemark; we indicate the number of crossings in the parenthesis.
		\textbf{\textsf{(a)}} \texttt{Global conGreedy+} (113), \textbf{\textsf{(b)}} \texttt{Local conGreedy+} (105), 
		\textbf{\textsf{(c)}} \texttt{Global conGreedy+ \& LS} (64), and
		\textbf{\textsf{(d)}} \texttt{Local conGreedy+ \& LS} (72).}
	\label{fig:sample-brinkmann}
\end{figure}
\footnotetext{A drawing of the graph can be found online: \url{https://en.wikipedia.org/wiki/Brinkmann_graph}.}

\subparagraph*{Computed Drawings.}
We provide in \Cref{fig:sample-wagner,fig:sample-brinkmann} sample drawings of two graphs together with the witness drawings computed by our different algorithms.
Generally, we can see that the DP and \texttt{Global conGreedy+} produce drawings with a more consistent arrangement of the vertices on the individual spines.
Comparing Figures~\ref{fig:sample-wagner}b and~f, we can observe that the obtained drawing from \texttt{Global conGreedy+ \& LS} is nearly identical to the one of the DP, confirming the observed low crossing numbers.
The consistency across different spines yields also a visually more pleasing drawing, as we can see in \Cref{fig:sample-brinkmann}.
We provide in \shortLong{the full version~\cite{ARXIV}}{\Cref{app:drawings}} more sample drawings computed by our algorithms and analyze below their running time and the quality of the produced drawings.

\subparagraph*{Running Times, Tradeoffs, and the Influence of the Width.}
All heuristics terminated on every instance within fifteen minutes.
The DP could solve only thirteen instances within the time limit.
To evaluate the crossing minimization performance of the heuristics on a larger set of instances, we ran the DP with a time limit of six hours, after which 21 of the 55 instances, all of width at most four, could be solved optimally.
For the remaining 34 instances, \texttt{Global conGreedy+} obtained four, \texttt{Local conGreedy+} three, \texttt{Global conGreedy+ \& LS} 18, and \texttt{Local conGreedy+ \& LS} 20 times a solution with the fewest crossings.

\Cref{fig:per-instance-plot} shows the relative number of crossings, compared to the best-performing algorithm, i.e., the DP where available and one of the heuristics otherwise, and running time for each algorithm and instance.
We group instances by width and sort them within each group based on the number of bags. 
{%
	In \shortLong{the full version~\cite{ARXIV}}{\Cref{fig:per-instance-plot_dp}}, %
	we provide a filtered version of \Cref{fig:per-instance-plot}, focusing on instances for which the DP terminated within six hours.}
Regarding the running time, we can observe in \Cref{fig:per-instance-plot}a that all heuristics without the local search step terminated within a second on every instance except the largest ones.
\begin{figure}
	\centering
	\includegraphics[page=1,width=\linewidth]{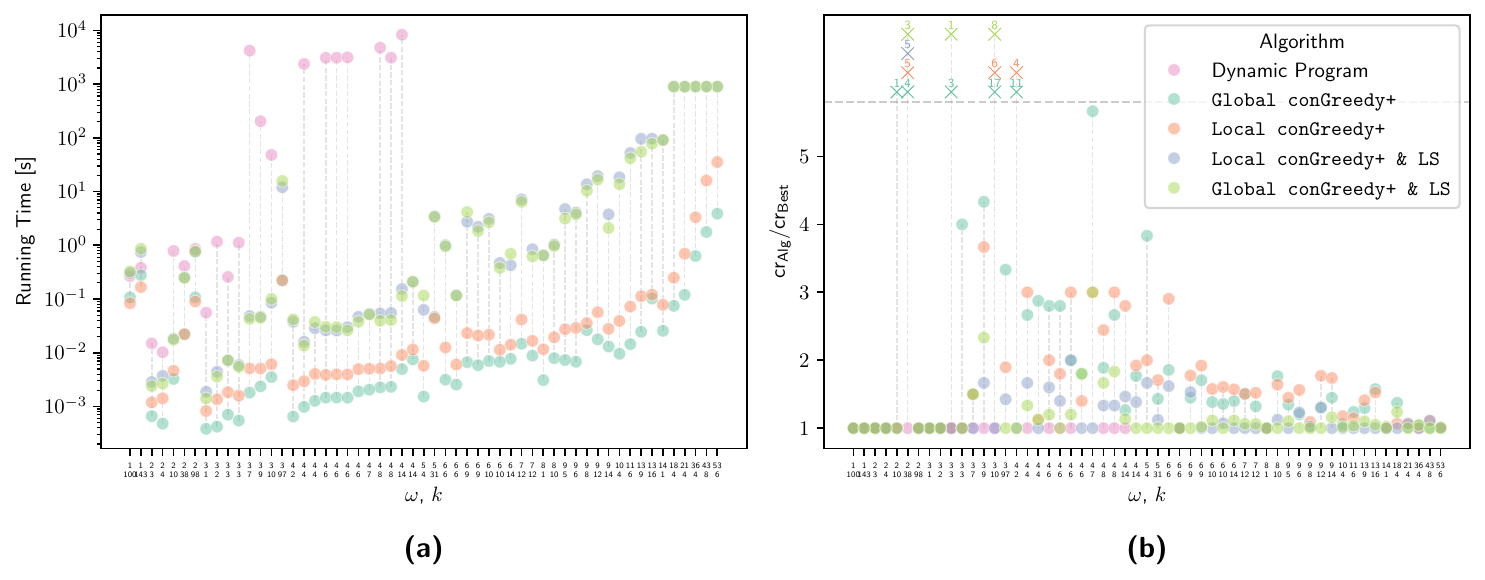}
	\caption{\textbf{\textsf{(a)}}~Wall-clock running time ($\log_{10}$-scale) of the heuristics and \textbf{\textsf{(b)}}~their optimality-ratio compared to the best found solution for different values of $\DecWidth$ and number of bags~$k$.
		We group instances based on $\DecWidth$, sort each group by~$k${%
			, and indicate on the $x$-axis the value of  $\DecWidth$ and $k$ in the first and second row, respectively.}    
		Vertical lines connect data points for one instance.}
	\label{fig:per-instance-plot}
\end{figure}
As expected, \texttt{Global conGreedy+} is the fasted heuristic, since it has to compute only one book drawing for the entire graph.
Regarding the optimality-ratio, \Cref{fig:per-instance-plot}b shows that for instances with larger width, \texttt{Global conGreedy+} seems to work slightly better than %
\texttt{Local conGreedy+}. %
We believe that computing a book drawing for each $G_i$ individually could yield situations where the obtained spine orders yield many \SymbolTTCrossing-crossings.
{%
	In particular, we suspect that spine orders that ``locally'' admit few crossings could ``globally'' introduce, for example, unexpected ``criss-cross'' \SymbolTTCrossing-crossings (recall \Cref{fig:crossing-counts}d).}
We see the strength of our DP for decompositions of width two or three.
For larger width, the high running time becomes impractical.
However, small-width instances are often simple enough that all heuristics can solve them to optimality in a few milliseconds.
Interestingly, for instances with very large treewidth, the benefit of local search drastically decreases: Despite performing movement operations for the entire fifteen minutes, the obtained solution is, if at all, only negligibly better then the corresponding heuristic without the local search.
However, for instances of moderate width, local search considerably improved the initial solution.

\begin{statelater}{perInstancePlotDP}
	\begin{figure}
		\includegraphics[page=1,width=\linewidth]{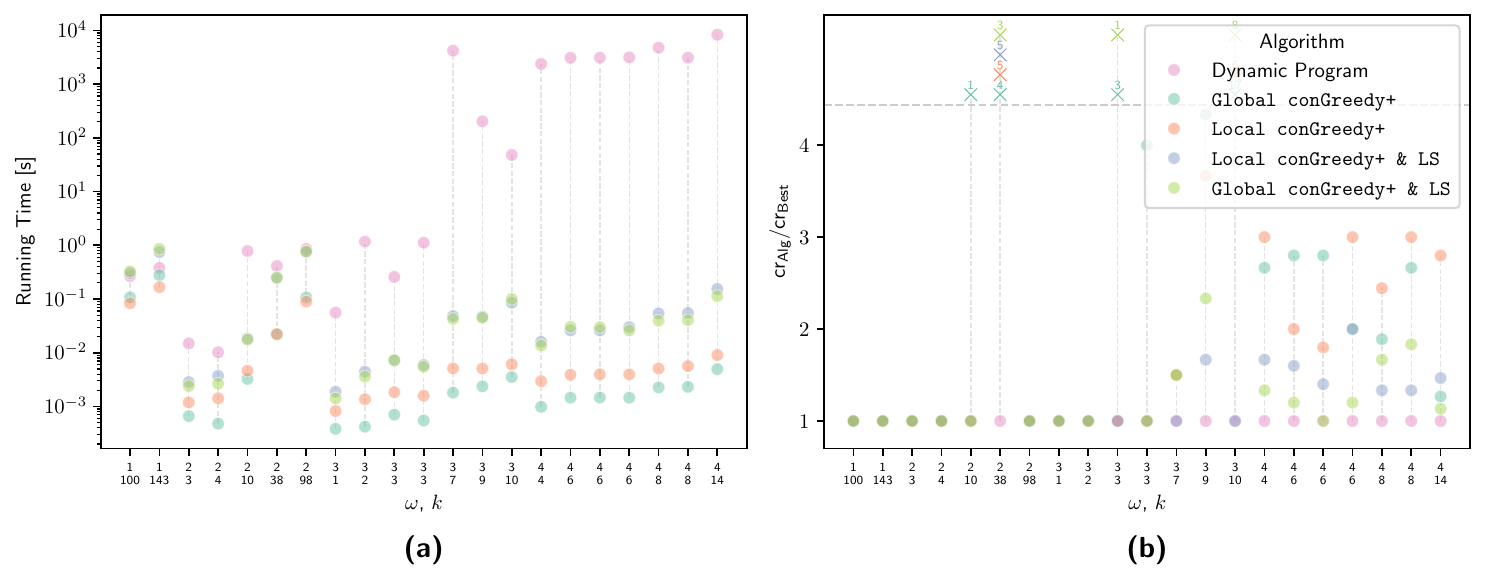}
		\caption{Filtered version of \Cref{fig:per-instance-plot}, showing only instances for which the DP terminated within six hours.}
		\label{fig:per-instance-plot_dp}
	\end{figure}
\end{statelater}
\section{Expert Evaluation}
\label{sec:user-study}
We conducted an expert evaluation to assess the usefulness and visual quality of the three drawing styles studied in this paper along with two drawing styles that visualize the graph~$G_i$, for each $V_i \in T$, as incidence matrix and with an embedding consistent to that of a drawing for the entire graph $G$, respectively; see the blue marked entries in \Cref{tab:design-space}.
The evaluation was designed as an online questionnaire, beginning with a brief introduction into witness drawings in general and a reminder of the definitions of tree decompositions and treewidth.\footnote{We remarked that participants should have a basic knowledge on these concepts.}
In the beginning, we informed the participants that they will see five different visualizations of a tree decomposition of a graph $G$ and showed them a node-link drawing and the incidence matrix of $G$; see \Cref{fig:questionnaire-sample}a in \Cref{app:user-study}.
We used the \emph{Krackhardt kite graph} from our dataset for the expert evaluation and showed the participants a small overview of the five different visualizations so that they could get familiar with them (similar to \Cref{fig:questionnaire-sample}b--f).
The linear drawing was created using our dynamic program, and the four remaining ones were created by hand by adapting the linear drawing while ensuring that the number of crossings remains comparable.
We showed to the participants one drawing after the other, in a large version and with the node-link drawing of $G$ next to it; see \Cref{fig:questionnaire-sample}b--f in \Cref{app:user-study} for the individual witness drawings.
The participants should rate on a 5-point Likert scale how hard or easy they consider performing {%
	the following} five different tasks with the help of the drawing.
{%
	\begin{description}
		\item[T1)] Verifying that each vertex is in at least one bag.
		\item[T2)] Verifying that for each edge there is a bag containing its endpoints.
		\item[T3)] Verifying that the bags containing a given vertex form a connected subtree.
		\item[T4)] Determining the width of the decomposition.
		\item[T5)] Mapping between the drawings inside the bags and the subdrawings of the graph.
	\end{description}
}
To balance learning effects, we used \emph{Latin Square} to determine the order in which the drawings were shown to the participants.
In the end, we asked the participants to rate the five drawings based on visual appeal and indicate the importance they assign to a set of properties such as the number of crossings, or verifiability of the tree decomposition.
The full questionnaire is available on OSF~\cite{SUPPLEMENTARY_MATERIAL}.

We invited participants of the \emph{International Symposium on Graph Drawing and Network Visualization} 2025 (\emph{GD'25}) to take part in the evaluation.
Additionally, we invited Bachelor's, Master's, and PhD students from collaborating research groups with a background in theoretical computer science and/or graph drawing. %
Overall, twelve persons participated in the evaluation (one Bachelor or Master student, ten PhD students, and one postdoctoral researcher or above).
Of these, two considered themselves proficient with tree decompositions and treewidth, eight had previously worked with these concepts and had also seen visual examples, and two participants indicated that they had heard the concepts before but never worked with them.
Eight participants indicated that they work in the area of graph drawing, five in computational geometry and parameterized algorithms, respectively, four in algorithm engineering, and one in information visualization and artificial intelligence each (multiple answers were possible).
Below, we analyze the participants' responses in more detail.

\subparagraph*{Results on Visualization Styles.}
Recall that we asked participants, for each drawing, how hard or easy they consider performing five different tasks.
\Cref{fig:study-drawing-styles} shows the considered tasks and the corresponding {%
	average scores} of the responses.
In the following, we refer to the tasks with T1--T5; also indicated in \Cref{fig:study-drawing-styles}.
\begin{figure}
	\centering
	\includegraphics[page=1,width=.95\linewidth]{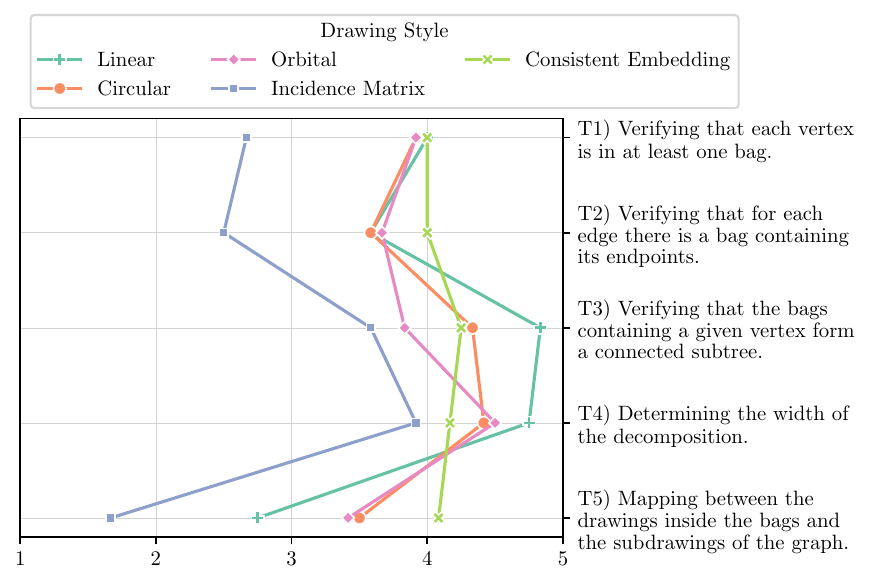}
	\caption{{%
			Average scores} of the five considered drawing styles. Participants were asked to rate on a 5-point-Likert scale how hard (1) or easy (5) they consider performing each of the tasks T1--T5.}
	\label{fig:study-drawing-styles}
\end{figure}
In general, we can observe that the visualization using incidence matrices {%
	was rated} considerably worse than the other styles for all tasks.
The remaining styles {%
	were rated} similarly for {%
	T1, T2, and T4}.
{%
	For T1, i.e., verifying that each vertex appears in at least one bag, no significant difference among the four drawing styles was perceivable.}
For~T2, i.e., verifying that each edge appears in at least one bag, the drawing style that uses the consistent embedding was preferred.
Interestingly, for T2, this style {%
	was also rated} notably better than the three drawing styles studied in this paper, although the circular and orbital drawings are, at least compared to linear drawings, visually closer to the style that uses a consistent embedding.
For verifying the last criterion, i.e., that bags containing a specific vertex form a connected subtree, the orbital drawing style {%
	was rated} worse than the other three styles.
We suspect that forcing the tracks to stay within the track routing area is the main reason for this outcome.
In particular, using a track routing area can result in visual clutter that can hinder the participants in following individual tracks; compare also \Cref{fig:study-drawing-styles}d to the other drawing styles in \Cref{fig:study-drawing-styles}.
On the other hand, we believe that a vertical alignment of vertices inside the disks such as in linear drawings facilitated the completion of this task, resulting in higher scores for the linear style for T3.
For T4, i.e., validating the width of the decomposition, %
the style that uses incidence matrices {%
	was rated the} worst. %
This is surprising since the size of the matrix should be a salient feature to detect and examine the largest bag.
With T5, we aimed at evaluating the recognizability of graph features in the drawing of the decomposition.
As expected, the style that uses consistent embeddings {%
	was rated} best for this task and the one that uses the incidence matrix by far the worst.
{%
	Observe that the matrix-based drawing was rated worst across nearly all tasks, which could also be attributed to a general unfamiliarity with matrix-style visualizations compared to node-link diagrams}. %

\begin{figure}
	\centering
	\includegraphics[page=1,width=\linewidth]{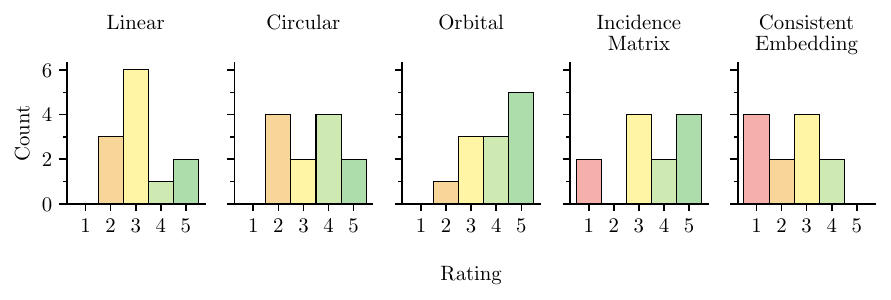}
	\caption{Summary of the ratings for the visual quality (1 is the worst score, 5 the best), indicating how often each style received which rating.}
	\label{fig:study-drawing-ratings}
\end{figure}
We also asked participants to rate the visual quality of the five drawings from 1 (worst) to 5 (best).
\Cref{fig:study-drawing-ratings} shows the results.
From top to bottom, the participants ranked the styles as follows: orbital drawing (4.00 on average), style based on incidence matrices (3.50), circular drawing (3.33), linear drawing (3.17), and the worst the style with the consistent embedding (2.33).
We can observe that orbital drawings and those using the incident matrix received, on average, the highest scores.
Thus, the perceived visual quality does not seem to correlate {%
	positively} with the perceived difficulty in performing tasks; compare \Cref{fig:study-drawing-styles,fig:study-drawing-ratings}.
In particular, recall that the matrix-based style {%
	was rated} worst across nearly all tasks, and the orbital style often {%
	was rated} worse than drawings using a linear, circular, or consistent arrangement of the vertices.
For orbital drawings, we believe that this outcome is mainly influenced by the track routing area; compare also circular and orbital drawings in \Cref{fig:study-drawing-ratings}.
The style with the consistent embedding received the lowest scores on average, despite {%
	receiving good rates} in the task-related questions; recall \Cref{fig:study-drawing-styles}.
Half of the participants selected the middle rating for linear drawings.

\subparagraph*{Results on the Importance of Properties.}
Finally, we asked participants to judge the importance of the properties listed in \Cref{fig:study-drawing-importance}.
\begin{figure}
	\centering
	\includegraphics[page=1,width=\linewidth]{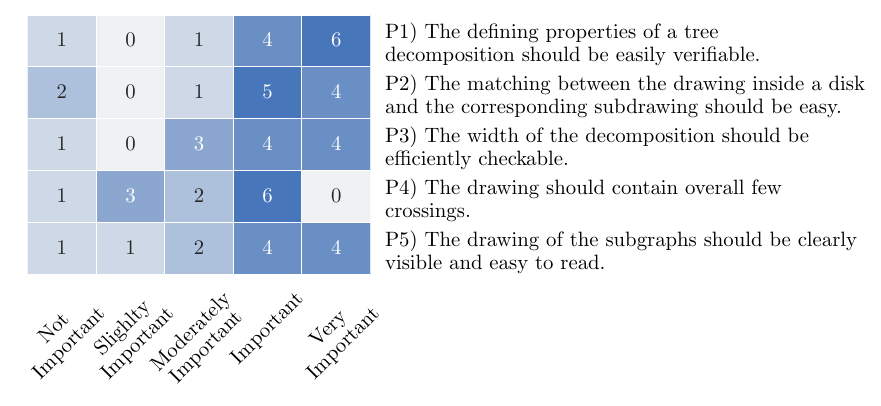}
	\caption{Summary of the indicated importance of the five properties. The numbers in the cells indicate how often each property received which importance score. {%
			We additionally color-code the numerical value using different shades of blue}.}
	\label{fig:study-drawing-importance}
\end{figure}
This time, we explicitly named the options, ranging from \emph{Not Important} (we associate to this a score of 1) to \emph{Very Important} (score~5).
In the following, we refer to the properties as P1--P5; also indicated in \Cref{fig:study-drawing-importance}.
Using above-mentioned scores, we obtain an average importance rating of 4.17 for P1, 3.75 for P2, 3.83 for P3, 3.08 for~P4, and 3.75 for P5.
{%
	Every property was rated by at least half of the participants} as important or {%
	very important}.
Since our drawings should serve as a visual certificate for bounded treewidth, P1 and P3 targeted the verifiability of the defining properties of a bounded-width tree decomposition.
Being able to verify that the decomposition is valid (P1) was considered important (or very important) by almost all participants.
Only for one participant this was not important at all.
However, for this person all properties except P3 were rated as unimportant.
Regarding the verifiability of the width of the decomposition, i.e., P3, one participant indicated that this was not important and three indicated that it was only moderately important.
Compared to P1, fewer thought that this was very important.
Overall, nine people indicated that P2 was (very) important, which is more than P3, P4, and P5, i.e., being able to verify the width of the decomposition, having few crossings in the drawing, and a readable drawing of the subgraphs $G_i$.
This, in particular, means that ensuring consistency across the drawings in the bags was indicated as more important than minimizing the crossings.
One participant even commented that verifying whether each edge is part of at least one bag, i.e., T2 in \Cref{fig:study-drawing-styles}, was generally tedious, but easier when the drawings made clear that the bags induced a cover of the entire graph.
Together with the responses to~P5, this suggests that, when creating witness drawings, one should prioritize the structure of the drawings in the bags over minimizing the number of crossings.
In general, the responses for P2 and P5 indicate that information faithfulness is important when designing witness drawings.

\subparagraph*{Concluding Discussion.}
Altogether, this expert evaluation partly confirms our design choices for witness drawings and the considered drawing styles.
Furthermore, the results indicate that witness drawings can aid users in verifying that a given tree decomposition is valid\footnote{In fact, one participant even correctly noted that one track edge was missing in a preliminary version of \Cref{fig:study-drawing-styles}e, which demonstrates a profound understanding of the semantics of our drawings.} and that the treewidth of a graph $G$ is bounded.
On the other hand, the expert evaluation revealed that crossing minimization, while being considered important for the readability in graph drawing~\cite{Pur.Wah.1997}, might not be the primary optimization criterion for the studied witness drawings.
Instead, we should aim for a consistent representation of the graphs across bags and with respect to a drawing of the entire graph $G$ (if there is one).
In particular, for such combined visualizations, i.e., drawings of the tree decomposition and $G$, we might face conflicting optimization criteria: Consistent drawings in the decomposition could result in a low-quality drawing of $G$.
Especially for decompositions of larger width, this optimization criterion needs further investigations.
For large-width decompositions, we can no longer rely on color as an additional visual variable to differentiate vertices in the bags, as also noted by one participant.
Hence, we may want to investigate additional visual encodings for such drawings.
Finally, participating experts also commented that the current visualizations are targeted on verifying that each edge is in at least one bag and every vertex induces a connected subtree, but less on the condition that every vertex is in at least one bag or the width of the decomposition.
This underlines the potential for exploring other, additional, design elements to convey these properties and/or interactive witness drawings.

\section{Concluding Remarks and Future Directions}
In this paper, we have initiated the study of visualizing tree and path decompositions of graphs to visually certify their treewidth and pathwidth, respectively.
{%
	We have explored the design space of such witness drawings and further examined three styles from an algorithmic perspective.}
Our experimental study shows that for decompositions of width at most three, we can compute exact crossing-minimal drawings in a few seconds.
For decompositions of medium to larger width, the heuristics combined with a local search step seem to be the preferred method.
We see the computation of a crossing-minimal drawing where we do not restrict the placement of the child bags or the starting angle inside the disk as a challenging open algorithmic problem.
{%
	Note that our exact algorithms can be adapted to handle a discrete set of candidate starting angles $\alpha$ but the problem of optimizing this parameter remains open.
}

{%
	Our expert evaluation confirms that the three studied drawing styles can aid in verifying that a given tree decomposition is valid.
	At the same time, the evaluation also shows that crossing minimization is not necessarily the primary optimization criterion to aim for.
	Consequently, we see the exploration of well-defined optimization criteria that result in a consistent appearance of the drawings across bags as an important direction for future work.
	Moreover, our description of the design space in \Cref{sec:taxonomy} should be seen as an invitation to develop further {%
		witness-drawing styles along with corresponding} exact and heuristic algorithms {%
		to compute them.
		The empirical and experimental evaluation of these new width-witness drawings is also an interesting open problem.
	}%
	In particular, we see a larger user study measuring actual task performance as a natural next step.%
}
Finally, designing and computing suitable witness drawings for other width parameters of graphs is an interesting{%
	, independent} direction for future work.

\section*{Acknowledgments}
We thank all participants for taking part in the expert evaluation.
{%
	Furthermore, we thank the anonymous reviewers for their valuable comments.}

\bibliography{references}

\appendix
\onlyShort{	
	\newpage
	\section{Omitted Proofs}
	\label{app:omitted-proofs}
	
	\theoremHardness*
	\label{thm:hardness*}
	\ptheoremHardness 
	\figureHardness
	
	\lemmaCorrectness*
	\label{lem:pathwidth-general-dp-correctness*}
	\plemmaCorrectness
	
	\theoremCircularDrawings*
	\label{thm:decomposition-circular*}
	\ptheoremCircularDrawings
	
	\newpage
	\section{Omitted Plots}
	\label{app:omitted-plots}
	
	\perInstancePlotDP
}

\newpage
\section{Sample Drawings}
\label[appendix]{app:drawings}
We provide in \Cref{fig:sample-dp,fig:sample-ls} further witness drawings computed by our algorithms.
\Cref{fig:sample-dp} shows four drawings computed by the DP.
\Cref{fig:sample-ls} shows the drawings for four decompositions computed by \texttt{Global conGreedy+ \& LS} and \texttt{Local conGreedy+ \& LS}.
We provide for all drawings the number of crossings in the parenthesis.
For further information on the respective graphs, we refer to the overview page of \emph{Sage}\footnote{\url{https://doc.sagemath.org/html/en/reference/graphs/sage/graphs/graph_generators.html}}, on which also the dataset is based on.\bigskip

\begin{figure}[bht]
	\centering
	\includegraphics[page=1]{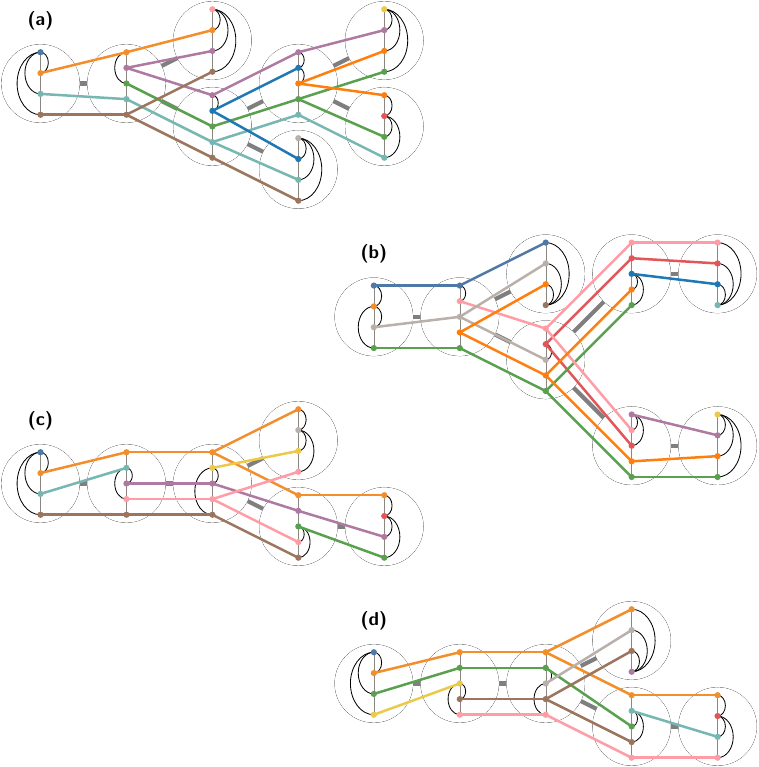}
	\caption{Sample drawings computed by our DP for decompositions of various graphs; we indicate the number of crossings in the parenthesis.
		\textbf{\textsf{(a)}} Bidiakis Cube (6), \textbf{\textsf{(b)}} Franklin Graph (9), 
		\textbf{\textsf{(c)}} Odd Graph 3 (5), and
		\textbf{\textsf{(d)}} Petersen Graph (5).}
	\label{fig:sample-dp}
\end{figure}

\begin{figure}[bhtp]
	\centering
	\includegraphics[page=1,height=\dimexpr\textheight-72.5pt\relax]{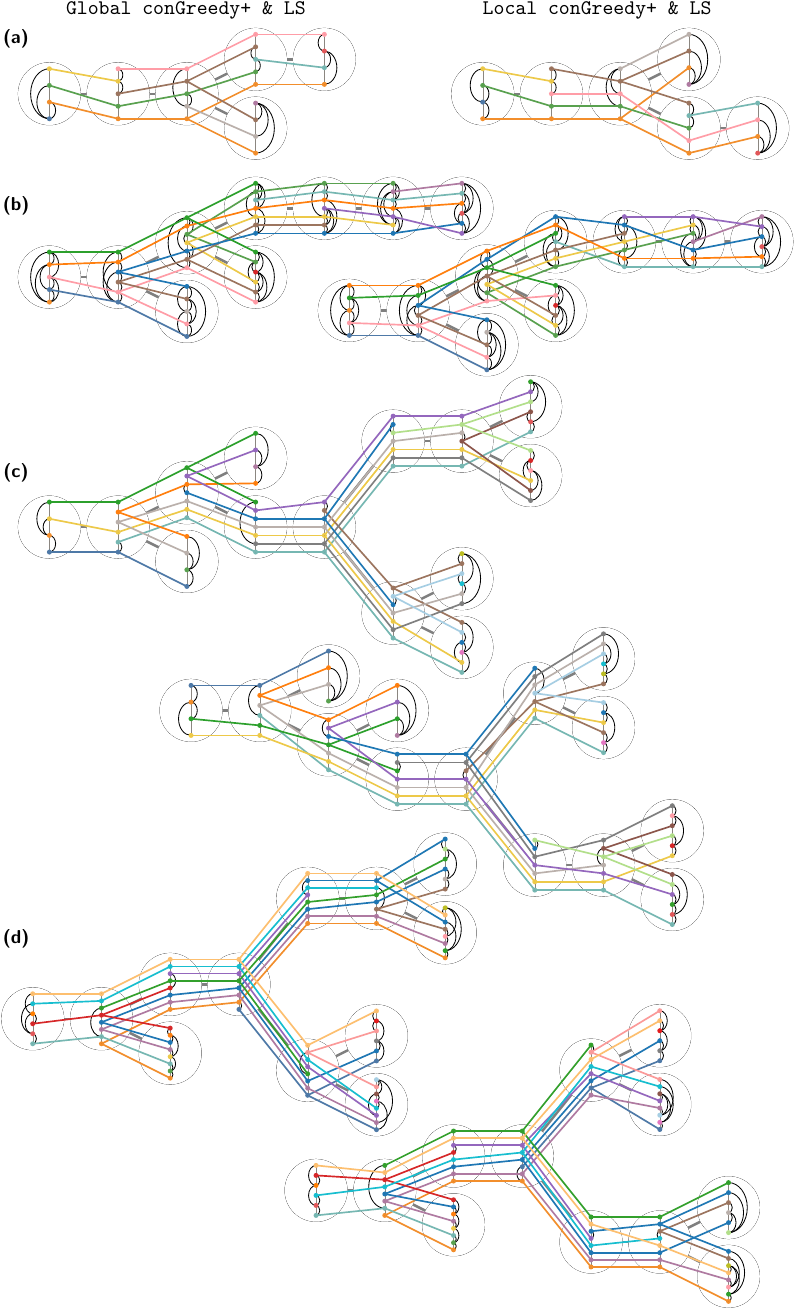}
	\caption{Sample drawings computed by the heuristics combined with the local search for decompositions of various graphs; we indicate the number of crossings in the parenthesis. A vertical bar (\texttt{Global conGreedy+ \& LS}$\vert$\texttt{Local conGreedy+ \& LS}) separates the number of crossings for the two algorithms.
		\textbf{\textsf{(a)}} Petersen Graph ($5\vert7$), \textbf{\textsf{(b)}} Poussin Graph ($52\vert51$), 
		\textbf{\textsf{(c)}} Nauru Graph ($39\vert35$), and
		\textbf{\textsf{(d)}} Coxeter Graph ($53\vert50$).}
	\label{fig:sample-ls}
\end{figure}

\newpage
\section{Drawings From the Expert Evaluation}
\label[appendix]{app:user-study}
\Cref{fig:questionnaire-sample} shows the drawings that we provided during the expert evaluation.
\Cref{fig:questionnaire-sample}a visualizes the graph both as node-link diagram and incidence matrix.
\Cref{fig:questionnaire-sample}b--f contain the five witness drawings that we evaluated.

\begin{figure}[bht]
	\centering
	\includegraphics[page=1]{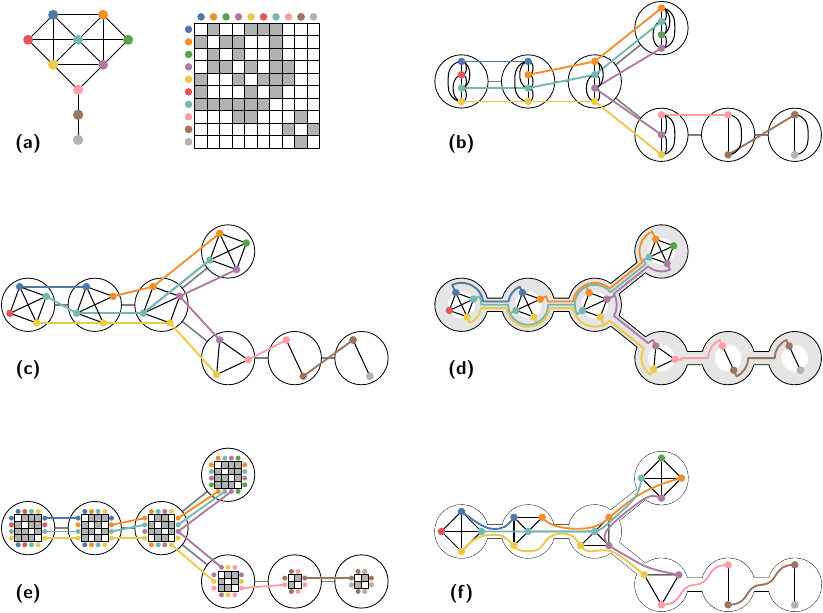}
	\caption{Drawings used in the expert evaluation.
		\textbf{\textsf{(a)}} Graph visualized as node-link diagram and incidence matrix, \textbf{\textsf{(b)}} linear drawing, 
		\textbf{\textsf{(c)}} circular drawing, \textbf{\textsf{(d)}} orbital drawing, 
		\textbf{\textsf{(e)}} incidence matrix, and
		\textbf{\textsf{(f)}} consistent embedding.}
	\label{fig:questionnaire-sample}
\end{figure}

\end{document}